\documentclass[12pt]{article}

\usepackage{amsfonts}
\usepackage{amssymb}
\usepackage{amsmath}
\usepackage{mathpazo}
\usepackage{hyperref}
\usepackage{multimedia}
\usepackage{color}

\usepackage{geometry}

\newtheorem{theorem}{Theorem}

\newtheorem{claim}{Claim}[section]

\newtheorem{corollary}{Corollary}[section]

\newtheorem{lemma}{Lemma}[section]

\newtheorem{proposition}{Proposition}[section]

\newenvironment{proof}[1][Proof]{\noindent\textbf{#1.} }{\ \rule{0.5em}{0.5em}}

\begin{document}

\title{Spreading Information via Social Networks: \\
An Irrelevance Result\thanks{%
The research reported here was supported by a grant from the National
Science Foundation (SES-2048806). We thank Nageeb Ali and Rahul Krishna for
helpful discussions.}}
\author{Yu Awaya\thanks{%
Department of Economics, University of Rochester, E-mail: yuawaya@gmail.com.}
\ and Vijay Krishna\thanks{%
Department of Economics, Penn State University, E-mail: vkrishna@psu.edu.}}
\date{February 5, 2024}
\maketitle

\begin{abstract}
An informed planner wishes to spread information among a group of agents in
order to induce efficient coordination---say the adoption of a new
technology with positive externalities. The agents are connected via a
social network. The planner informs a seed and then the information spreads
via the network. While the structure of the network affects the rate of
diffusion, we show that the rate of adoption is the \emph{same} for all 
\emph{acyclic} networks.
\end{abstract}

\section{Introduction}

Policymakers often wish to inform the public about various policies and
technologies---a tax credit, a new seed variety, a new digital payment
system---so that the public will avail of or adopt these. The simplest way
to disseminate such information is to just broadcast a public service
announcement (PSA) on television, radio and other media. But there have been
some doubts about the effectiveness of PSAs. For various reasons, people may
pay more attention to information coming from friends and neighbors rather
than mass media.\footnote{%
Banerjee et. al (2023) observe this in the "field" and provide some
behavioral explanations why this might be the case.
\par
{}} Thus in many circumstances it is better to "seed" the information to a
few individuals and then let it spread naturally via the existing social
network. Of course, how quickly information diffuses depends on the network.
At one extreme, information will spread very quickly in a "star"
network---where one individual, say $1$, is directly connected to all others
who are directly connected only to $1$. At the other extreme, it will spread
very slowly in a "line" network---where individual $1$ is connected only to
individual $2,$ who is connected only to one other, say $3$, etc. See Figure %
\ref{F: Starline}. 
%
%
%
%
%
%
%
%
%
%
%
%
%
%
%
%
%
%
%

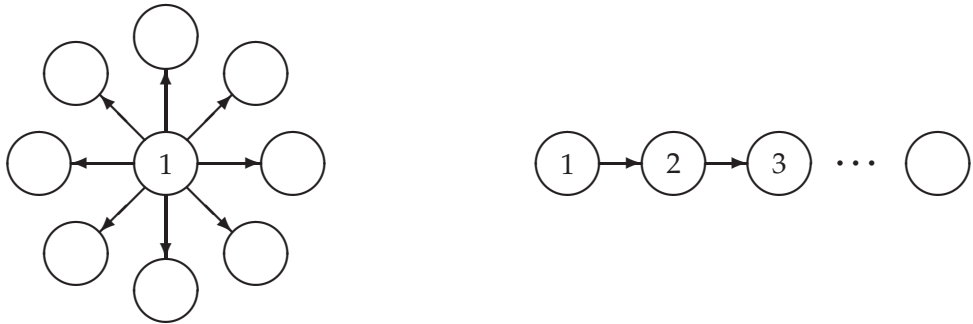
\begin{figure}[t]
\unitlength 2pt
\begin{picture}(198,60)
\linethickness{0.5pt}
\small
\thicklines

\put(50,30){\circle {12}}
\put(50,30){\makebox(0,0)[cc]{$1$}}

\put(50,54){\circle {12}}
\put(50,36){\vector(0,1){12}}

\put(67,47){\circle {12}}
\put(54,34.5){\vector(1,1){8.5}}

\put(74,30){\circle {12}}
\put(56,30){\vector(1,0){12}}

\put(67,13){\circle {12}}
\put(54,25.5){\vector(1,-1){8.5}}

\put(50,6){\circle {12}}
\put(50,24){\vector(0,-1){12}}

\put(33,13){\circle {12}}
\put(46,25.5){\vector(-1,-1){8.5}}

\put(26,30){\circle {12}}
\put(44,30){\vector(-1,0){12}}

\put(33,47){\circle {12}}
\put(46,34.5){\vector(-1,1){8.5}}

\put(126,30){\circle {12}}
\put(126,30){\makebox(0,0)[cc]{$1$}}
\put(132,30){\vector(1,0){8}}

\put(146,30){\circle {12}}
\put(146,30){\makebox(0,0)[cc]{$2$}}
\put(152,30){\vector(1,0){8}}

\put(166,30){\circle {12}}
\put(166,30){\makebox(0,0)[cc]{$3$}}

\multiput(178,30)(3,0){3}{\makebox(0,0)[cc]{\circle* {1}}}

\put(196,30){\circle {12}}

\end{picture}
\caption{Star and  Line Networks}
\label{F: Starline}
\end{figure}%


In this paper, we ask a different question. Suppose that the information
concerns the benefits of a new technology---say, a new digital payment
system---and the policymaker wishes to get the public to adopt the system.
In many such situations, there are significant positive
externalities---adopting a new digital payment system is useful only if
other people do so as well.\footnote{%
See Crouzet, Gupta and Mezzanotti (2023) for evidence of such extenalities
in the adoption of a digital wallet following the Indian demonetization in
2016.} Instead of asking how the network structure affects the rate of
diffusion, we ask how it affects the rate of \emph{adoption}.

Our main finding is that in the class of \emph{acyclic} networks, the
structure of the network and how it is seeded is irrelevant---the rate of
adoption is the \emph{same} for all such networks. In particular, the
adoption rate when information diffuses quickly via the star network is the
same as when it diffuses slowly via the line network. So while the structure
of the network affects both the speed at which information is diffused and,
as we will see, its quality, it does not affect the prospects of efficient
coordinated behavior.

To illustrate our findings, we begin with an example.

\subsection{Example}

A new technology is of uncertain value---it may or may not be useful/viable.
There are three agents who must simultaneously decide whether or not to
adopt the new technology at a cost $c<1$ per person. The gross payoff to an
agent is $\$1$ if and only if the technology turns out to be useful and 
\emph{all} three agents adopt the technology; otherwise, the gross payoff is
zero. Thus, adoption has positive externalities. Let $\rho \in \left(
0,1\right) $ be the prior probability that the technology is useful.

A planner, agent $0$, knows whether or not the technology is useful, and if
it is, sends a message to the agents. If it is not useful, no message is
sent. Thus, anyone who gets the message is sure that the technology is
useful.%
%
%
%
%
%
%
%
%
%
%
%
%
%
%
%
%
%
\begin{figure}[t]
\unitlength 2pt
\begin{picture}(160,70)(-12,-5)
\linethickness{0.5pt}
\small
\thicklines

\put(-6,30){\circle {12}}
\put(-6,30){\makebox(0,0)[cc]{$1$}}
\put(0,30){\line(1,0){10}}

\put(16,30){\circle {12}}
\put(16,30){\makebox(0,0)[cc]{$2$}}
\put(22,30){\line(1,0){10}}

\put(38,30){\circle {12}}
\put(38,30){\makebox(0,0)[cc]{$3$}}

\put(16,5){\makebox(0,0)[cc]{(a)}}

\put(75,54.8){\circle {16}}
\put(75,54.8){\makebox(0,0)[cc]{$0$}}
\put(75,46.8){\vector(0,-1){10.8}}

\put(75,30){\circle {12}}
\put(75,30){\makebox(0,0)[cc]{$1$}}
\put(81,30){\vector(1,0){10}}

\put(97,30){\circle {12}}
\put(97,30){\makebox(0,0)[cc]{$2$}}
\put(103,30){\vector(1,0){10}}

\put(119,30){\circle {12}}
\put(119,30){\makebox(0,0)[cc]{$3$}}

\put(97,5){\makebox(0,0)[cc]{(b)}}

\put(176,54.8){\circle {16}}
\put(176,54.8){\makebox(0,0)[cc]{$0$}}
\put(176,46.8){\vector(0,-1){10.8}}

\put(154,30){\circle {12}}
\put(154,30){\makebox(0,0)[cc]{$1$}}
\put(170,30){\vector(-1,0){10}}

\put(176,30){\circle {12}}
\put(176,30){\makebox(0,0)[cc]{$2$}}
\put(182,30){\vector(1,0){10}}

\put(198,30){\circle {12}}
\put(198,30){\makebox(0,0)[cc]{$3$}}

\put(176,5){\makebox(0,0)[cc]{(c)}}

\end{picture}
\caption{Seeding the Line Network}
\label{F: Three Line}
\end{figure}
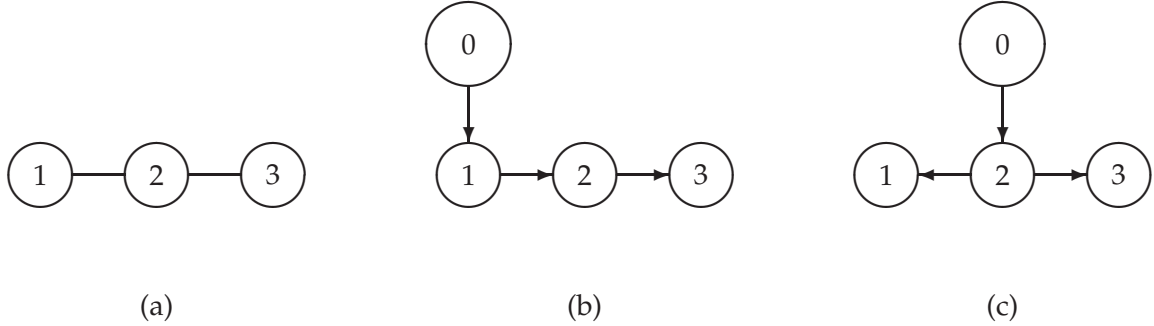%

The three agents are part of single connected social network. Specifically,
they are arranged along a line as in Figure \ref{F: Three Line} (a). Agents
can receive messages from and pass these along to their neighbors. Message
transmission is imperfect, however---at every stage there is a small
probability $\varepsilon >0$ that a message that is sent to a neighbor is
lost. Thus, if the planner sends a message to $1,$ there is only a
probability $1-\varepsilon $ that $1$ will in fact get the message. If $1$
receives the message and sends it on to $2$, then there is only a
probability $1-\varepsilon $ that $2$ will get the message and so on.
Transmission losses occur independently across links.

\paragraph{A1. Seeding the network via $1.$}

Here, if the technology is useful, the planner sends a message to $1,$ which
if received, is sent to $2,$ which if received, is then sent to $3$ (this is
depicted by the arrows in Figure \ref{F: Three Line} (b)). We claim that if
the cost $c\leq \left( 1-\varepsilon \right) ^{2}$, then every agent who is
informed---gets a message---adopts.\footnote{%
For our purposes it is not necessary to specify the exact strategy---that
is, what an agent does if she does not get the message. A detailed
specification of the strategies is in Section \ref{S: Game}.} And if $%
c>\left( 1-\varepsilon \right) ^{2}$, then no agent, informed or not, adopts.

\subparagraph{Case 1: $c\leq \left( 1-\protect\varepsilon \right) ^{2}.$}

First, consider agent $1$ and suppose the others adopt if informed. If $1$
gets a message, then she knows that the technology is useful and so her only
worry is whether all other agents got the message as well. Since messages
are only passed along the line, the probability that agents $2$ and then $3$
are also informed, and so will adopt, is just $\left( 1-\varepsilon \right)
^{2}$ which is greater than the cost.\footnote{%
Since the gross payoff if the technology is useful and everyone adopts is $1$
and $0$ otherwise, this probability is also the gross expected payoff.}
Thus, if informed, it is optimal for agent $1$ to adopt.

Next, consider agent $2$ and as above, suppose the others adopt if informed.
If $2$ gets a message, then she knows that $1$ also got the message and that 
$3$ got the message with probability $1-\varepsilon $. Thus, it is optimal
for agent $2$ to adopt as long as $c\leq 1-\varepsilon ,$ a weaker
requirement than that for agent $1.$

Finally, consider agent $3$ and again suppose others adopt if informed. If $%
3 $ gets a message, then she knows for sure that $1$ and $2$ also got the
message and so she is willing to adopt for all $c\leq 1$.

Thus, if $c\leq \left( 1-\varepsilon \right) ^{2}$ there is an equilibrium
in which every agent adopts if she gets the message. The probability that
all agents adopt the technology when it is useful is just the probability
that the message reaches $3$, that is, $\left( 1-\varepsilon \right) ^{3}.$

\subparagraph{Case 2: $c>\left( 1-\protect\varepsilon \right) ^{2}.$}

In this case, the \emph{unique} equilibrium is one in which no agent ever
adopts.

To see why, note that if $1$ is uninformed, the probability that she assigns
to the event that the technology is useful is small---it is of order $%
\varepsilon $. This is because the only way this can happen is if the
message from the planner to $1$ was lost, an $\varepsilon $ probability
event. When $\varepsilon $ is small, this probability is smaller than the
cost. This means that it is dominated for an uninformed agent $1$ to adopt.

Now from the argument above, $2$ knows that $1$ will not adopt if
uninformed. If $2$ does not get a message, her belief that $1$ is informed
is also of order $\varepsilon $ because the event that $1$ is informed while 
$2$ is not can occur only if the message from $1$ to $2$ was lost, again an $%
\varepsilon $ probability event. This means that it is (iteratively)
dominated for an uninformed $2$ to adopt.

Now $3$ knows that $1$ and $2$ will not adopt if uninformed. If $3$ does not
get a message, then for similar reasons as above, her belief that both $1\ $%
and $2$ got a message is again of order $\varepsilon .$ So it is
(iteratively) dominated for an uninformed $3$ to adopt.

Thus we have argued that it is \emph{iteratively} dominated for every agent
to adopt if she does not get a message.

Now suppose agent $1$ is informed. At best, the other agents will adopt only
if informed and the chance of this is $\left( 1-\varepsilon \right) ^{2}$
and since $c$ exceeds this, it is optimal for $1$ to not adopt even when
informed. Thus, agent $1$ will never adopt.

But now if agent $1$ never adopts, it is optimal for other agents to never
adopt as well.

\paragraph{A2. Seeding the network via $2.$}

Seeding the network via $1$ seems inefficient since the information has to
travel from $1$ to $2$ and then from $2$ to $3.$ Suppose instead that the
planner sends a message to the agent who is "central," that is, $2$. This
message, if received by $2$, is forwarded to $1$ and $3$ simultaneously
(this is depicted by the arrows in Figure \ref{F: Three Line} (c)). We claim
that even though this seems like a better way to disseminate information,
the prospects for efficient coordination are the \emph{same} as when $1$ is
the seed. Again, if $c\leq \left( 1-\varepsilon \right) ^{2},$ then every
informed agent adopts. And if $c>\left( 1-\varepsilon \right) ^{2}$, then no
agent ever adopts.

\subparagraph{Case 1: $c\leq \left( 1-\protect\varepsilon \right) ^{2}.$}

Here an informed agent $2$'s belief that the others will get the message is
just $\left( 1-\varepsilon \right) ^{2}$ which is greater than the cost.
Thus, upon getting a message it is optimal for agent $2$ to adopt if the
others are doing so.

The same calculation for $1$ shows that $1$ is willing to adopt as long as $%
c\leq 1-\varepsilon .$ This is because $1$ knows that $2$ is informed for
sure and that $3$ is informed with probability $1-\varepsilon .$ Thus, it is
optimal for an informed agent $1$ to adopt as long as $c\leq 1-\varepsilon $%
, a weaker requirement than for $2.$ The same is true for agent $3$ since $1$
and $3$ are symmetrically placed.

Thus, it is an equilibrium for every agent to adopt if she is informed. When
the technology is useful, the probability that the message reaches all the
agents is just $\left( 1-\varepsilon \right) ^{3},$ the same as in scenario
A1.

\subparagraph{Case 2: $c>\left( 1-\protect\varepsilon \right) ^{2}.$}

Again, the unique equilibrium is one in which no agent ever adopts.

To see why, note that if $2$ does not get a message, the probability that
she assigns to the event that the technology is useful is again of order $%
\varepsilon $ (it is the same as that assigned by $1$ when $1$ was the seed
in Scenario A1). Since $c$ is greater than this probability, this implies
that it is dominated for an informed $2$ to adopt$.$

Now $1$ knows that $2$ will not adopt if uninformed. As above, if $1$ does
not get a message, her belief that $2$ got a message is also of order $%
\varepsilon .$ This belief is again smaller than the cost. This means that
it is (iteratively) dominated for an uninformed $1$ to adopt. As before, the
same is true for $3$ since $1$ and $3$ are symmetrically placed.

Now suppose agent $2$ gets a message. The chance that the other agents will
adopt is at most $\left( 1-\varepsilon \right) ^{2}$ and since $c$ exceeds
this, it is optimal for $2$ to not adopt even when she gets a message. Thus,
agent $2$ will not adopt whether or not she is informed.

But now if agent $2$ never adopts, it is optimal for $1$ and $3$ to never
adopt as well.

Thus, we see that in this example, seeding information to a more "central"
agent---with more neighbors---does not improve the prospects for
coordination.

\paragraph{B. Broadcasting.}

An alternative is to bypass the social network entirely and "broadcast" the
message---it is sent privately to all agents simultaneously (as in Figure %
\ref{F: Three Broadcast}).%
%
%
%
%
%
%
%
%
%
\begin{figure}[t]
\unitlength 2pt
\begin{picture}(160,50)(-10,20)
\linethickness{0.5pt}
\small
\thicklines

\put(103.4,49.2){\vector(1,-1){14.7}}
\put(90.6,49.2){\vector(-1,-1){14.7}}

\put(97,54){\circle {16}}
\put(97,54){\makebox(0,0)[cc]{$0$}}
\put(97,46){\vector(0,-1){10}}

\put(72,30){\circle {12}}
\put(72,30){\makebox(0,0)[cc]{$1$}}

\put(97,30){\circle {12}}
\put(97,30){\makebox(0,0)[cc]{$2$}}

\put(122,30){\circle {12}}
\put(122,30){\makebox(0,0)[cc]{$3$}}

\end{picture}
\caption{Broadcast}
\label{F: Three Broadcast}
\end{figure}
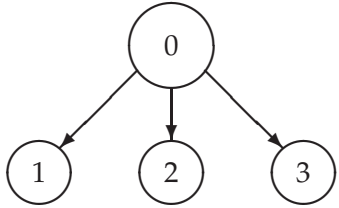
Again, with probability $\varepsilon $, the message to any agent $i$ is lost
and so not heard by the agent. Lost messages occur independently across
agents, each with probability $\varepsilon >0$.

It seems intuitive that directly broadcasting is a better method of
dissemination than letting the information trickle from agent to agent. In
particular, with broadcasting the probability that any agent gets the
information is $1-\varepsilon $ and so is the same for all agents. In
contrast, when the network is seeded via $1$, the probability that agent $2$
gets the message is $\left( 1-\varepsilon \right) ^{2}$ and the probability
that $3$ gets the message is $\left( 1-\varepsilon \right) ^{3}$.

We will show that even though broadcasting provides better information about
the usefulness of the technology, when it comes to engendering efficient
coordination, it is \emph{equivalent} to either of the indirect methods of
dissemination A1 and A2.

\subparagraph{Case 1: $c\leq \left( 1-\protect\varepsilon \right) ^{2}.$}

Again, in this case there is an equilibrium in which every agent who gets
the message adopts the technology. This follows from the fact that for any
informed agent $i$, the probability that the other two agents are also
informed is $\left( 1-\varepsilon \right) ^{2}.$

When the technology is useful, the probability that all agents get the
message is just $\left( 1-\varepsilon \right) ^{3},$ the \emph{same} as that
when the network is seeded via $1$ or $2.$

\subparagraph{Case 2: $c>\left( 1-\protect\varepsilon \right) ^{2}.$}

In this case, again the unique equilibrium is again one in which no one ever
adopts. If agent $i$ does not get a message, her belief that the technology
is useful is again of order $\varepsilon $ which is less than $c$. This
means it is dominated for an uninformed agent to adopt. If agent $i$ does
get a message, the probability that the other two agents will adopt is at
most $\left( 1-\varepsilon \right) ^{2}$ and since $c>\left( 1-\varepsilon
\right) ^{2}$, agent $i$ will not adopt even if informed. This means that
the unique equilibrium is for all agents to never adopt.\bigskip

Thus, in this example we see that how information is
disseminated---indirectly via seeding the network or directly sending it to
each agent---does not affect the prospects for coordination.

In this paper we show that there is nothing special about the example. As in
the example, we study information dissemination in social networks without
cycles---that is, trees. Informally stated, our main result is:\footnote{%
The formal results (Theorem \ref{T: main} below) applies not only to a
single tree but also to networks which are collections of disjoint trees---%
\emph{forests}.}\medskip

\emph{The prospects of efficient coordination are the same} \emph{no matter
how information is disseminated---it is independent both of the structure of
the tree and how it is seeded.}\medskip

Our main result says that if the goal of the planner is to induce efficient
coordinated action---adopting a new technology or product---then the tree
structure is irrelevant. Why is this? Efficient coordination requires not
only that agents be informed about the fundamental uncertainty---whether or
not the technology is useful---but also be informed whether other agents are
informed. It is easy to see that one tree network may be better than another
in conveying information about fundamentals. For instance, in the example
above, let us compare broadcasting (Scenario B) to seeding $1$ and then
letting the information trickle from $1$ to $2$ to $3$ (Scenario A1).
Broadcasting is better than seeding in that every agent has better
information about the fundamentals with broadcasting than seeding. But
seeding is better than broadcasting in revealing whether others are
informed---for instance, if agent $3$ gets a message, she knows that $1$ and 
$2$ did so as well. \footnote{%
A precise result along these lines is Proposition \ref{Blackwell} below.}

Rather than consider a particular game---like the adoption game in the
example---we derive and phrase our results by using the language of
approximate common knowledge. Thus our main result says that the extent of
approximate common knowledge is independent of the structure of the tree
network and how it is seeded. The close connection between approximate
common knowledge and equilibrium behavior in games is well-known (see
Monderer and Samet, 1989, Kajii and Morris, 1997 and Oyama and Takahashi,
2020).

\subsection{Related literature}

The question of diffusion in social networks appears in many
contexts---infectious diseases, product awareness, plans for a revolt, etc.
In most of these situations, the planner is interested in the affecting the
speed and extent of diffusion---either slowing it down in the case of
disease or speeding it up in the other cases. Recently, the question has
drawn the attention of development economists who are interested in
conveying information about various policy initiatives and has been studied
in various contexts---microfinance (Banerjee et al. 2013), immunizations
(Banerjee et al. 2019), planting techniques (Beaman et al. 2021),
demonetization (Banerjee et al. 2023)---by conducting randomized controlled
trials (RCTs). But these issues are not confined to developing countries.
Chetty et al. (2013) find that whether or not people optimally avail of the
earned income tax credit---a large US government transfer program---depends
on the neighborhood they live in.

One of the findings of this line of research is that spreading the
information via existing social networks may be superior to "broadcasting"
the information via media (Banerjee et al. 2023).

One is then naturally led to the question of how best to "seed" the
information by conveying it to a few key agents who spread it via the social
network. Clearly, if one is interested in spreading the information quickly
and widely, the information should be seeded via agents that are
well-connected---that is, central players. But identifying who is central is
daunting task in any reasonable sized network. First, one has to determine
the network---a difficult task itself---and second, to find the central
players in the network. The latter problem is known to be computationally
hard.

In a very interesting paper, Akbarpour, Malladi and Saberi (2023) have
argued that instead of finding the optimal seed, it is better to choose
multiple seeds randomly. The argument is that even if the information is
seeded to non-central agents it will find its way to those that are central
anyway---by definition, the central players are well-connected.

In all of this work, the focus is on the speed at which information spreads
as well as the extent of diffusion. Implicit in this is the assumption that
once an agent is informed, he/she will automatically adopt the new
technology or avail of the policy initiative. This may be the case if the
costs and benefits of adoption do not depend on whether others are doing so
as well. But many new technologies/products are subject to complementarities
in adoption/consumption---that is, network externalities. Crouzet et al.
(2023) document the presence of such externalities in the adoption of a new
digital payment platform in India. Naturally, adopting a digital payment
platform is useful only if others adopt it as well. Such externalities are
also a key feature of the theoretical diffusion model of Sadler (2020).

Our paper departs from the focus on the speed of diffusion. Rather, we are
interested in the likelihood that the new technology---subject to network
externalities---will be adopted once the information has spread. Put another
way, to what extent will the public be able to coordinate adoption? As is
well-known, efficient coordination requires that people know not only
whether or not the digital wallet works and is safe---known as first-order
uncertainty---but whether others know this as well and whether others know
that others know, etc.---higher-order uncertainty. The importance of
considering higher-order uncertainty is the main lesson of Rubinstein's
(1989) E-mail game who shows that it can be a major cause of coordination
failure.\footnote{%
Coles and Shorrer (2012) show that the extreme coordination failure in
Rubinstein's two-player E-mail game can be mitigated in multi-player games
where communication takes place in a hub-and-spoke network. In De Jaegher
(2015) higher-order information is directly communicated to the agents.}

Field experiments by Gottlieb (2016) point to the importance of reducing
higher-order uncertainy in elections in Mali. An interesting finding along
the same lines is reported by George, Gupta and Neggers (2019). They measure
how broadcasting---via text messages---the criminal records of candidates
affects vote shares in Indian elections. In one treatment of an RCT, the
texts just report the criminal records of the candidates and the effect on
vote shares is negligible. In a second treatment, the texts not only report
the criminal records but also say that others are also receiving the same
text. Now the effect on vote shares is measurable---there is roughly a 3\%
shift in votes. Thus, it seems that people are willing to change their
voting behavior only if they think that others will do so as well. In other
words, reducing higher-order uncertainty matters.

The change in focus away from the speed of diffusion to efficient
coordination is the key to our result. Once efficient coordination is the
goal, we find that in acyclic networks, the network structure and who is the
seed becomes irrelevant. In such networks, there is no need to ascertain the
exact structure or the optimal person to choose as the seed.\bigskip\ 

\subparagraph{Organization of the Paper}

The remainder of the paper is organized as follows. The next section
outlines the model as well as the terminology of approximate common
knowledge. Section \ref{S: Main} then derives the main result that in all
acyclic networks the extent of approximate common knowledge is the same.
Section \ref{S: Game} studies the technology adoption game from the
Introduction, as well as other games, and shows that the approximate common
knowledge results of Section \ref{S: Main} have exact counterparts
concerning equilibria of these games. Our results rely on the assumptions
that (a) the network is acyclic; and (b) each tree in the network has a
single seed. In Section \ref{S: Other} we show that the irrelevance result
does not extend if these assumptions are relaxed. We show by example that
there are circumstances in which a cycle can make the situation worse in
terms of the adoption rate. Also, there are circumstances in which a single
seed is better than multiple seeds. Finally, we also consider the
possibility of randomly chosen seeds. We show that choosing seeds at random
can, in some cases, improve the situation. Appendices A and B contain proofs.

\section{Model}

There is an uncertain state of nature $\theta =g$ or $b$ with prior
probabilities $\rho \in \left( 0,1\right) $ and $1-\rho ,$ respectively. A
planner who knows $\theta $ wishes to convey this information to a set of
agents $\mathcal{I}=\left\{ 1,2,...,I\right\} $---the public. The planner
will be labeled as agent $0$.

The agents in $\mathcal{I}$ constitute the nodes of a social network which
is either a \emph{tree }$T$---an undirected connected graph without
cycles---or a disjoint union $T^{1}\cup T^{2}\cup ...\cup T^{R}$ of trees,
that is, a \emph{forest}. Let $F=\left( T^{1},T^{2},...,T^{R}\right) $
denote the forest.\footnote{%
For formal definitions of these and other terms, we refer the reader to the
excellent book by Jackson (2008).}

In state $g$, the planner sends a private message to a \emph{single} node $%
s^{r}$ in each tree $T^{r}$---the \emph{seed} of $T^{r}.$ Let $s=\left(
s^{1},s^{2},...,s^{R}\right) $ denote a \emph{seeding} of the forest. The
message then spreads through each tree as follows.

Fix a particular tree, say $T^{1}$, and let agent $1$ be the seed. If the
seed gets a message, she forwards it to each $i$ in the set of her
neighbors, denoted by $\mathcal{N}\left( 1\right) $. Each of $1$'s neighbors 
$i\in \mathcal{N}\left( 1\right) $ then forwards the message to each of his
neighbors $j\in \mathcal{N}\left( i\right) $ \emph{except} $1.$ Each $j\in 
\mathcal{N}\left( i\right) \backslash \left\{ 1\right\} $ then forwards the
message to each of her neighbors in $\mathcal{N}\left( j\right) $ except $i$
and so on.

In this manner, information spreads throughout the tree. Notice that because
(a) there are no cycles in the underlying undirected network; and (b) no
agent sends a message back to the person she received a message from, it is
the case that now the tree becomes \emph{directed}---there is single
direction of flow of information from seed to all other nodes. Formally, for
every node $i,$ there is a unique agent that immediately precedes $i$ in the
tree and is the only source of information for $i.$ Given a forest $F$ and a
seeding $s,$ let $\mathcal{T}\left( F,s\right) $ denote the resulting
directed tree with the planner, agent $0$, as the root. We will refer to $%
\mathcal{T}$ as the (directed) \emph{information tree}.

The top panel of Figure \ref{F: Seeding} depicts a forest consisting of two
undirected trees. The other two panels show how the choice of different
seeds results in different directed trees. In each case, the arrows depict
the flow of information.%
\begin{figure}[p]
\unitlength 2pt
\begin{picture}(198,270)(-6,-25)
\linethickness{0.5pt}
\small
\thicklines

\put(37,230){\circle {12}}
\put(37,230){\makebox(0,0)[cc]{$1$}}

\put(63,230){\circle {12}}
\put(63,230){\makebox(0,0)[cc]{$2$}}
\put(66,224.5){\line(2,-3){6}}

\put(34,224.5){\line(-2,-3){6}}
\put(40,224.5){\line(2,-3){6}}
\put(43,230){\line(1,0){14}}

\put(24,210.5){\circle {12}}
\put(24,210.5){\makebox(0,0)[cc]{$3$}}

\put(50,210.5){\circle {12}}
\put(50,210.5){\makebox(0,0)[cc]{$4$}}

\put(76,210.5){\circle {12}}
\put(76,210.5){\makebox(0,0)[cc]{$5$}}

\put(124,230){\circle {12}}
\put(124,230){\makebox(0,0)[cc]{$6$}}
\put(144,230){\line(-1,0){14}}

\put(150,230){\circle {12}}
\put(150,230){\makebox(0,0)[cc]{$7$}}

\put(176,230){\circle {12}}
\put(176,230){\makebox(0,0)[cc]{$8$}}
\put(156,230){\line(1,0){14}}

\put(150,210.5){\circle {12}}
\put(150,210.5){\makebox(0,0)[cc]{$9$}}
\put(150,224){\line(0,-1){7.5}}

\put(100,162){\circle {16}}
\put(100,162){\makebox(0,0)[cc]{$0$}}
\put(92.8,158.4){\vector(-2,-1){50}}
\put(107.2,158.4){\vector(3,-2){37.7}}

\put(37,130){\circle {12}}
\put(37,130){\makebox(0,0)[cc]{$1$}}
\put(34,124.5){\vector(-2,-3){6}}
\put(40,124.5){\vector(2,-3){6}}
\put(43,130){\vector(1,0){14}}

\put(24,110.5){\circle {12}}
\put(24,110.5){\makebox(0,0)[cc]{$3$}}

\put(50,110.5){\circle {12}}
\put(50,110.5){\makebox(0,0)[cc]{$4$}}

\put(63,130){\circle {12}}
\put(63,130){\makebox(0,0)[cc]{$2$}}
\put(66,124.5){\vector(2,-3){6}}

\put(76,110.5){\circle {12}}
\put(76,110.5){\makebox(0,0)[cc]{$5$}}

\put(124,130){\circle {12}}
\put(124,130){\makebox(0,0)[cc]{$6$}}
\put(144,130){\vector(-1,0){14}}

\put(150,130){\circle {12}}
\put(150,130){\makebox(0,0)[cc]{$7$}}

\put(176,130){\circle {12}}
\put(176,130){\makebox(0,0)[cc]{$8$}}
\put(156,130){\vector(1,0){14}}

\put(150,110.5){\circle {12}}
\put(150,110.5){\makebox(0,0)[cc]{$9$}}
\put(150,124){\vector(0,-1){7.5}}

\put(100,62){\circle {16}}
\put(100,62){\makebox(0,0)[cc]{$0$}}
\put(92.8,58){\vector(-1,-3){14}}
\put(107.2,58){\vector(2,-3){14.8}}

\put(37,30){\circle {12}}
\put(37,30){\makebox(0,0)[cc]{$1$}}
\put(34,24.5){\vector(-2,-3){6}}
\put(40,24.5){\vector(2,-3){6}}
\put(57,30){\vector(-1,0){14}}

\put(24,10.5){\circle {12}}
\put(24,10.5){\makebox(0,0)[cc]{$3$}}

\put(50,10.5){\circle {12}}
\put(50,10.5){\makebox(0,0)[cc]{$4$}}

\put(63,30){\circle {12}}
\put(63,30){\makebox(0,0)[cc]{$2$}}

\put(76,10.5){\circle {12}}
\put(76,10.5){\makebox(0,0)[cc]{$5$}}
\put(73,15.5){\vector(-2,3){6}}

\put(124,30){\circle {12}}
\put(124,30){\makebox(0,0)[cc]{$6$}}
\put(130,30){\vector(1,0){14}}

\put(150,30){\circle {12}}
\put(150,30){\makebox(0,0)[cc]{$7$}}

\put(176,30){\circle {12}}
\put(176,30){\makebox(0,0)[cc]{$8$}}
\put(156,30){\vector(1,0){14}}

\put(150,10.5){\circle {12}}
\put(150,10.5){\makebox(0,0)[cc]{$9$}}
\put(150,24){\vector(0,-1){7.5}}

\end{picture}
\caption{Undirected Forest to Directed Tree by Seeding}
\label{F: Seeding}
\end{figure}%

Messages can be lost, however. If $i$ forwards a message to her neighbor $j$%
, then there is a probability $\varepsilon >0$ that the message is lost and
not received by $j.$ Conditional on $i$ being informed, the losses of $i$'s
messages to her neighbors are independent. Thus, if $i$ sends messages to
her neighbors $j$ and $k,$ then the probability that both will receive the
message is $\left( 1-\varepsilon \right) ^{2}.$ The same is true for
messages from the planner to a seed---the probability that in state $g$, the
seed receives the message is also $1-\varepsilon .$

Note that if the message from $i$ to $j$ is lost, then $j$ cannot forward it
to anyone and the flow of information to all the nodes that succeed $j$
stops.

In state $b$, no messages are sent by the planner and so there is no flow of
information.

Note that if $i$ receives a message, then he knows for sure that (1) the
state of nature is $g$; and (2) all agents $j$ along the unique path from
the seed to his immediate predecessor also received a message.

\paragraph{Common beliefs}

Let $x_{i}\in \left\{ y,n\right\} $ denote the information available to $%
i\in \mathcal{I}$, where $x_{i}=y$ ("yes") denotes that $i$ received a
message and $x_{i}=n$ ("no") denotes that $i$ did not.

The set of \emph{states of the world }is%
\[
\Omega =\left\{ g,b\right\} \times \left\{ y,n\right\} ^{I} 
\]%
A state of the world $\omega \in \Omega ,$ then determines both the
fundamental state, $g$ or $b,$ as well as which of the agents are informed.
For instance, the state of the world $\left( g,y,y,n\right) $ is one where
the fundamental is $g,$ agents $1$ and $2$ receive the message while agent $%
3 $ does not.

Different network structures and seedings lead to different probability
distributions on $\Omega .$ For instance, if three agents are arranged in a
line as in Figure \ref{F: Three Line}, then state $\left( g,n,y,y\right) $
is impossible if $1$ is the seed, while it has a positive probability of
occurring if $2$ is the seed. In what follows, we will fix the forest $%
F=\left( T^{1},T^{2},...,T^{R}\right) $ and a seeding $s=\left(
s^{1},s^{2},...,s^{R}\right) $ so also the resulting information tree $%
\mathcal{T}.$ All probabilities will be calculated using the resulting
probability distribution $\mathbb{P}_{\mathcal{T}}$ over $\Omega .$ To avoid
notational clutter, in what follows, we will temporarily suppress the
dependence of $\mathbb{P}_{\mathcal{T}}$ and other objects on $\mathcal{T}.$

Following Monderer and Samet (1989), given any event $E\subseteq \Omega $
and probability $p$, the event $B_{i}^{p}\left( E\right) \subseteq \Omega $
consists of states of the world $\omega $ in which $E$ is $p$\emph{-believed}
by $i$ given the information $x_{i}\left( \omega \right) \in \left\{
y,n\right\} $ available to her in state $\omega .$ Formally,%
\[
B_{i}^{p}\left( E\right) =\left\{ \omega \in \Omega :\mathbb{P}\left[ E\mid
X_{i}=x_{i}\left( \omega \right) \right] \geq p\right\} 
\]%
In other words, in any state $\omega \in B_{i}^{p}\left( E\right) ,$ $i$
assigns probability exceeding $p$ to the event $E$ given her information $%
x_{i}\left( \omega \right) .$ We write 
\[
B^{p}\left( E\right) =\cap _{i\in \mathcal{I}}B_{i}^{p}\left( E\right) 
\]%
as the set of states in which $E$ is $p$\emph{-}believed by all the agents.

Now for $\ell =1,2,...$ define $B^{p,\ell }$ recursively by 
\[
B^{p,\ell }\left( E\right) =B^{p}\left( B^{p,\ell -1}\left( E\right) \right) 
\]%
where $B^{p,0}\left( E\right) =E$ and finally,%
\[
C^{p}\left( E\right) =\cap _{\ell \geq 1}B^{p,\ell }\left( E\right) 
\]%
Thus, $C^{p}\left( E\right) $ is the set of states of the world in which $E$
is \emph{common }$p$\emph{-believed}. In other words, (i) everyone assigns
probability exceeding $p$ to the event $E,$ and also (ii) assigns
probability exceeding $p$ to the event that everyone assigns probability
exceeding $p$ to the event $E,$ and also (iii) assigns probability exceeding 
$p$ to the event that everyone assigns probability exceeding $p$ to the
event that everyone assigns probability exceeding $p$ to the event $E,$ and
so on.

Note that $B_{i}^{p}$ is a monotone mapping, that is, $E\subseteq E^{\prime
} $ implies that $B_{i}^{p}\left( E\right) \subseteq B_{i}^{p}\left(
E^{\prime }\right) .$ The same is then true of $B^{p,\ell }$ and $C^{p}.$
Also, if for some $\ell $ it is the case that $B^{p,\ell +1}\left( E\right)
=B^{p,\ell }\left( E\right) ,$ then $C^{p}\left( E\right) =B^{p,\ell }\left(
E\right) .$ Thus, $C^{p}\left( E\right) $ is a fixed point of $B^{p}$.

When $p=1,$ $C^{1}\left( E\right) $ is the set of states in which the event $%
E$ is commonly known. For $p$ close to $1,$ $C^{p}\left( E\right) $ is the
set of states in which $E$ is approximately commonly known.

We emphasize once again that since the probability distribution $\mathbb{P}$
over states depends on the underlying information tree $\mathcal{T}$, the
sets $B_{i}^{p}\left( E\right) $, $B^{p}\left( E\right) $ and $C^{p}\left(
E\right) $ also depend on $\mathcal{T}.$ Later when we want to make this
dependence explicit, we will write $\mathbb{P}_{\mathcal{T}}$ and $C_{%
\mathcal{T}}^{p}\left( E\right) ,$ for instance.

\subparagraph{Notation}

Let 
\[
G=\left\{ \omega \in \Omega :\theta =g\right\} 
\]%
be the set of states of the world with $\theta =g$ and let%
\[
Y_{i}=\left\{ \omega \in \Omega :x_{i}\left( \omega \right) =y\right\} 
\]%
consist of states in which $i$ is informed (gets a message). Since messages
are sent only in state $g,$ $Y_{i}\subset G.$ Finally, let%
\[
Y^{\ast }=\cap _{i\in \mathcal{I}}Y_{i} 
\]%
be the set of states in which every agent is informed. Since $y$ is
conclusive evidence that the $\theta =g,$ in fact, $Y^{\ast }$ consists of a
single state $\omega ^{\ast }=\left( g,y,y,...,y\right) .$

\section{Irrelevance of structure\label{S: Main}}

Rather than considering a specific game, say the technology adoption game
from the Introduction, we begin by showing that the set of common $p$%
-beliefs does not depend on the network or its seeding. It is well-known
that the degree of approximate common knowledge is a fundamental determinant
of equilibrium behavior in incomplete information games (Kajii and Morris,
1997).

Our main result is that the extent of approximate common knowledge is
independent of the underlying information tree $\mathcal{T}=\left(
F,s\right) .$ It depends only on the number of agents $I$, the prior $\rho $
and the error probability $\varepsilon $.

\begin{theorem}
\label{T: main}For any $p$, the event $C_{\mathcal{T}}^{p}\left( G\right) $
in which $G$ is common $p$-believed does not depend on the information tree $%
\mathcal{T}.$ Moreover, the probability $\mathbb{P}_{\mathcal{T}}[C_{%
\mathcal{T}}^{p}\left( G\right) ]$ does not depend on $\mathcal{T}$ either.
\end{theorem}

\subsection{Proof of Theorem \protect\ref{T: main}}

The proof of Theorem \ref{T: main} is divided into two parts---when the
error probability $\varepsilon $ is small and when it is large.

Let agent $1$ be a seed of some tree in the forest and note that the
probability of $G$ given that $1$ is uninformed is 
\begin{equation}
\Pr \left[ G\mid N_{1}\right] =\frac{\rho \varepsilon }{1-\rho \left(
1-\varepsilon \right) }  \label{Pr[G given N1]}
\end{equation}%
From Lemma \ref{L: Ystar given Yk}, the probability that all other agents
are informed given that $1$ is informed is%
\begin{equation}
\Pr \left[ Y^{\ast }\mid Y_{1}\right] =\left( 1-\varepsilon \right) ^{I-1}
\label{Pr[Ystar given Y1]}
\end{equation}%
Note that these probabilities are the same for \emph{any} seed of \emph{any}
tree in the forest since all seeds receive information directly from the
planner. Thus we simply write $\Pr $ to denote these rather than $\mathbb{P}%
_{\mathcal{T}}.$

Let $\overline{\varepsilon }$ be the unique value of $\varepsilon $ that
equates $\Pr \left[ G\mid N_{1}\right] $ and $\Pr \left[ Y^{\ast }\mid Y_{1}%
\right] $. Such a value exists and is unique since $\Pr \left[ G\mid N_{1}%
\right] $ is an increasing function of $\varepsilon $ while $\Pr \left[
Y^{\ast }\mid Y_{1}\right] $ is a decreasing function.

\subsubsection{Small $\protect\varepsilon $}

When $\varepsilon <\overline{\varepsilon },$ it is the case that

\begin{equation}
\Pr \left[ G\mid N_{1}\right] =\frac{\rho \varepsilon }{1-\rho \left(
1-\varepsilon \right) }<\left( 1-\varepsilon \right) ^{I-1}=\Pr \left[
Y^{\ast }\mid Y_{1}\right]  \label{def epsilon bar}
\end{equation}%
We then have

\begin{proposition}
\label{P: small}If $0<\varepsilon <\overline{\varepsilon },$ then for any
information tree $\mathcal{T}$,%
\[
C_{\mathcal{T}}^{p}\left( G\right) =\left\{ 
\begin{tabular}{ll}
$\Omega $ & if $p\leq \Pr \left[ G\mid N_{1}\right] $ \\ 
$Y^{\ast }$ & if $\Pr \left[ G\mid N_{1}\right] <p\leq \Pr \left[ Y^{\ast
}\mid Y_{1}\right] $ \\ 
$\varnothing $ & if $p>\Pr \left[ Y^{\ast }\mid Y_{1}\right] $%
\end{tabular}%
\right. 
\]
\end{proposition}

\paragraph{Proof}

We consider each range of $p$'s separately.

\subparagraph{Case 1: $p\leq \Pr \left[ G\mid N_{1}\right] .$}

In this case, $p$ is so low that even an uninformed agent $1$ assigns
greater probability than $p$ to $G.$

Lemma \ref{L: G given Nk} now implies that for any agent $i$ in the forest,
seed or not, 
\[
\Pr \left[ G\mid N_{1}\right] \leq \Pr \left[ G\mid N_{i}\right] 
\]%
and so for all $i,$ $p\leq \Pr \left[ G\mid N_{i}\right] .$ This means that
every agent, informed or not, assigns a probability of at least $p$ to $G.$
Formally, 
\[
B_{i}^{p}\left( G\right) =Y_{i}\cup N_{i}=\Omega 
\]%
and since $B^{p}\left( G\right) =\cap _{i\in \mathcal{I}}B_{i}^{p}\left(
G\right) ,$ 
\[
B^{p}\left( G\right) =\Omega 
\]%
But since everyone assigns probability $1$ to $\Omega ,$ it follows that $%
C^{p}\left( G\right) =\Omega .$

\subparagraph{Case 2: $\Pr \left[ G\mid N_{1}\right] <p\leq \Pr \left[
Y^{\ast }\mid Y_{1}\right] .$}

This case is broken up into two steps.

\emph{Step 1}: $\Pr \left[ G\mid N_{1}\right] <p$ implies that $C^{p}\left(
G\right) \subseteq Y^{\ast }.$

To show this step we will argue that for any agent $k,$ $C^{p}\left(
G\right) \cap N_{k}=\varnothing .$ In other words, the event that $G$ is
common $p$-believed cannot include any state in which an agent is uninformed.

Consider the unique path from $0$ to $k$ and suppose (after renaming, if
necessary) that this path consists of agents $1,2,...,k$ such that the
direct predecessor of $k$ is $k-1.$ Note that $0$ is the direct predecessor
of $1.$

Then since $p>\Pr \left[ G\mid N_{1}\right] ,$ $B_{1}^{p}\left( G\right)
\cap N_{1}=\varnothing .$ But since $C^{p}\left( G\right) \subseteq
B^{p}\left( G\right) \subseteq B_{1}^{p}\left( G\right) $, it is also the
case that 
\begin{equation}
C^{p}\left( G\right) \cap N_{1}=\varnothing  \label{CpG int N1}
\end{equation}%
This is because if an uninformed agent $1$ does not assign probability $p$
to $G,$ then the event that $G$ is common $p$-believed cannot include any
state in which $1$ is uninformed.

Now from Lemma \ref{L: Yk-1 given Nk}, $\Pr \left[ Y_{1}\mid N_{2}\right]
<\Pr \left[ G\mid N_{1}\right] $ which is less than $p.$ So $B_{2}^{p}\left(
Y_{1}\right) \cap N_{2}=\varnothing .$ Next (\ref{CpG int N1}) implies that $%
C^{p}\left( G\right) \subseteq Y_{1}$ and since $B_{2}^{p}$ is a monotone
mapping, $B_{2}^{p}\left( C^{p}\left( G\right) \right) \subseteq
B_{2}^{p}\left( Y_{1}\right) .$ Finally, since $C^{p}\left( G\right) $ is a
fixed point of $B^{p}$, $C^{p}\left( G\right) =B^{p}\left( C^{p}\left(
G\right) \right) \subseteq B_{2}^{p}\left( C^{p}\left( G\right) \right) $
and so $C^{p}\left( G\right) \subseteq B_{2}^{p}\left( Y_{1}\right) .$ This
implies that%
\[
C^{p}\left( G\right) \cap N_{2}=\varnothing 
\]%
Proceeding in this way we see that for all agents $j$ along the path $%
1,2,...,k-1,$ $C^{p}\left( G\right) \cap N_{j}=\varnothing $ and so 
\[
C^{p}\left( G\right) \cap N_{k}=\varnothing 
\]%
In other words, the event that $G$ is common $p$-believed cannot include any
state in which $k$ is uninformed.

Since $k$ was arbitrary, we have shown that 
\begin{eqnarray*}
C^{p}\left( G\right) &\subseteq &\cap _{i\in \mathcal{I}}Y_{i} \\
&=&Y^{\ast }
\end{eqnarray*}

\emph{Step 2}: $p\leq \Pr \left[ Y^{\ast }\mid Y_{1}\right] $ implies that $%
Y^{\ast }\subseteq C^{p}\left( G\right) .$

Since $p\leq \Pr \left[ Y^{\ast }\mid Y_{1}\right] ,$ Lemma \ref{L: Ystar
given Yk} implies that for all $i,$ $\Pr \left[ Y^{\ast }\mid Y_{1}\right]
<\Pr \left[ Y^{\ast }\mid Y_{i}\right] ,$ we have that for all $i,$ 
\[
B_{i}^{p}\left( Y^{\ast }\right) =Y_{i} 
\]%
and taking intersections over $i,$ $B^{p}\left( Y^{\ast }\right) =\cap
_{i\in \mathcal{I}}Y_{i}=Y^{\ast }$ and so%
\[
C^{p}\left( Y^{\ast }\right) =Y^{\ast } 
\]

Now since $Y^{\ast }\subseteq G,$ and the $C^{p}$ operator is monotone, $%
C^{p}\left( Y^{\ast }\right) \subseteq C^{p}\left( G\right) $ and so%
\[
Y^{\ast }\subseteq C^{p}\left( G\right) 
\]

\subparagraph{Case 3: $p>\Pr \left[ Y^{\ast }\mid Y_{1}\right] .$}

Now $p$ is so high that $B_{1}^{p}\left( Y^{\ast }\right) =\varnothing $ and
so%
\[
C^{p}\left( Y^{\ast }\right) =\varnothing 
\]%
as well.

From Step 1 of Case 2, we already know that $C^{p}\left( G\right) \subseteq
Y^{\ast }$ and so%
\[
C^{p}\left( G\right) \subseteq C^{p}\left( Y^{\ast }\right) =\varnothing 
\]

This completes the proof. $\blacksquare $

\subsubsection{Large $\protect\varepsilon $}

When $\varepsilon \geq $ $\overline{\varepsilon },$ it is the case that%
\[
\Pr \left[ G\mid N_{1}\right] \geq \Pr \left[ Y^{\ast }\mid Y_{1}\right] 
\]%
We then have

\begin{proposition}
\label{P: large}If $\varepsilon \geq \overline{\varepsilon },$ then for any
information tree,%
\[
C_{\mathcal{T}}^{p}\left( G\right) =\left\{ 
\begin{tabular}{ll}
$\Omega $ & if $p\leq \Pr \left[ G\mid N_{1}\right] $ \\ 
$\varnothing $ & if $p>\Pr \left[ G\mid N_{1}\right] $%
\end{tabular}%
\right. 
\]
\end{proposition}

\paragraph{Proof}

\subparagraph{Case 1: $p\leq \Pr \left[ G\mid N_{1}\right] .$}

Here the proof is the same as in Case $1$ of Proposition \ref{P: small}.

\subparagraph{Case 2: $p>\Pr \left[ G\mid N_{1}\right] .$}

As in Step 1 of Case 2 in the proof of Proposition \ref{P: small},%
\[
C^{p}\left( G\right) \subseteq Y^{\ast } 
\]%
Now since $\Pr \left[ Y^{\ast }\mid Y_{1}\right] \leq \Pr \left[ G\mid N_{1}%
\right] <p$, the probability that $1$ assigns to $Y^{\ast }$ is less than $p$
and so $B_{1}^{p}\left( Y^{\ast }\right) =\varnothing .$ It now follows that 
\[
C^{p}\left( Y^{\ast }\right) =\varnothing 
\]

This completes the proof. $\blacksquare \bigskip $

Propositions \ref{P: small} and \ref{P: large} prove the first part of
Theorem \ref{T: main} since they show that $C_{\mathcal{T}}^{p}\left(
G\right) $ depends only on $\Pr \left[ G\mid N_{1}\right] $ and $\Pr \left[
Y^{\ast }\mid Y_{1}\right] ,$ both probabilities that are independent of the
information tree $\mathcal{T}=\left( F,s\right) .$

The second part of Theorem \ref{T: main} now follows as a simple consequence
of the two propositions.

\begin{proposition}
The probability $\mathbb{P}_{\mathcal{T}}[C_{\mathcal{T}}^{p}\left( G\right)
]$ does not depend on the information tree $\mathcal{T}$.
\end{proposition}

\begin{proof}
Propositions \ref{P: small} and \ref{P: large} show that the $C_{\mathcal{T}%
}^{p}\left( G\right) $ does not depend on the structure of the forest and
how it is seeded. It depends only on $\rho ,$ $I$ and $\varepsilon .$

When $C_{\mathcal{T}}^{p}\left( G\right) =\Omega $ or $\varnothing $, the
probabilities are obviously $1$ or $0$, respectively. When $C_{\mathcal{T}%
}^{p}\left( G\right) =Y^{\ast },$ Lemma \ref{L: Forest Prob} implies that
the probability is simply $\left( 1-\varepsilon \right) ^{I}.$
\end{proof}

\subsection{An informativeness perspective}

Some intuition for the irrelevance result can be gleaned by comparing
different information trees using the informativeness criterion of Blackwell
(1951).

Consider two information trees $\mathcal{T}=\left( F,s\right) $ and $%
\mathcal{T}^{\prime }=\left( F^{\prime },s^{\prime }\right) .$ Let $d\left(
i\right) $ denote the number of links between $i$ and the root, agent $0$,
in the information tree $\mathcal{T}$ and let $d^{\prime }\left( i\right) $
denote the analogous number in $\mathcal{T}^{\prime }.$

It is then natural to say that $\mathcal{T}$ \emph{diffuses information
faster} than $\mathcal{T}^{\prime }$, written $\mathcal{T}\succeq _{dif}%
\mathcal{T}^{\prime }$, if there is a permutation of the names of the agents 
$\pi :\mathcal{I}\rightarrow \mathcal{I}$ such that for each $i\in \mathcal{I%
},$ $d\left( i\right) \leq d\left( \pi \left( i\right) \right) .$

We will say that $\mathcal{T}$ is \emph{first-order} more informative than $%
\mathcal{T}^{\prime },$ written $\mathcal{T\succeq }_{FO}\mathcal{T}^{\prime
},$ if there is a permutation $\pi $ such that for each $i\in \mathcal{I},$ $%
i$'s information about $G$ versus $\Omega \diagdown G$ in $\mathcal{T}$ is
Blackwell more informative than $\pi \left( i\right) $'s information about $%
G $ versus $\Omega \diagdown G$ in $\mathcal{T}^{\prime }.$

Similarly, we will say that $\mathcal{T}$ is \emph{second-order }more
informative than $\mathcal{T}^{\prime }$, written $\mathcal{T\succeq }_{SO}%
\mathcal{T}^{\prime },$ if there is a permutation $\pi $ such for each $i\in 
\mathcal{I},$ $i$'s information about $Y^{\ast }$ versus $\Omega \diagdown
Y^{\ast }$ in $\mathcal{T}$ is Blackwell more informative than $\pi \left(
i\right) $'s information about $Y^{\ast }$ versus $\Omega \diagdown Y^{\ast
} $ in $\mathcal{T}^{\prime }.$ The terminology reflects the fact that $%
Y^{\ast }$ is the event that all agents know that the state is $G.$

The following proposition shows that while the diffusion ordering $\succeq
_{dif}$ is the same as the first-order $\mathcal{\succeq }_{FO}$ ranking,
the second-order $\mathcal{\succeq }_{Y^{\ast }}$ ranking runs in the
opposite direction. If $\mathcal{T}$ is better than $\mathcal{T}^{\prime }$
in conveying first-order information, it is worse that $\mathcal{T}^{\prime
} $ in conveying second-order information (and vice versa).

\begin{proposition}
\label{Blackwell}For any two information trees $\mathcal{T}$ and $\mathcal{T}%
^{\prime },$ 
\[
\text{(i) }\mathcal{T}\succeq _{dif}\mathcal{T}^{\prime }
\text{ if and only if }
\mathcal{T}\succeq _{FO}\mathcal{T}^{\prime } 
\]%
and%
\[
\text{(ii) }\mathcal{T}\succeq _{FO}\mathcal{T}^{\prime }\text{ if and only
if }\mathcal{T}\preceq _{SO}\mathcal{T}^{\prime } 
\]
\end{proposition}

\begin{proof}
First, given any permutation $\pi ,$ note that $d\left( i\right) \leq
d\left( \pi \left( i\right) \right) $ if and only if $i$'s information about 
$G$ is Blackwell superior to $\pi \left( i\right) $'s information. This is
because Lemma \ref{L: G given Nk} implies that $\Pr \left[ G\mid N_{i}\right]
\leq \Pr \left[ G\mid N_{\pi \left( i\right) }\right] $ whereas $\Pr \left[
G\mid Y_{i}\right] =\Pr \left[ G\mid Y_{\pi \left( i\right) }\right] =1$
since any $Y_{j}$ is conclusive evidence that $\theta =g.$ Thus, $d\left(
i\right) \leq d\left( \pi \left( i\right) \right) $ if and only if $i$'s
posterior beliefs about $G$ are a mean-preserving spread of $\pi \left(
i\right) $'s beliefs about $G$ versus $\Omega \diagdown G.$\footnote{%
In experiments with two "states," $G$ and $\Omega \diagdown G,$ and two
"signals," $Y_{i}$ and $N_{i},$ this is sufficient for ranking the
information in terms of the Blackwell criterion.} (i) now follows
immediately.

Second, $d\left( i\right) \leq d\left( \pi \left( i\right) \right) $ if and
only if $\pi \left( i\right) $'s second-order information is Blackwell
superior to $i$'s information. In the latter case, the Blackwell experiment
is well-defined since the agents have a common prior about $Y^{\ast }$ given
by Lemma \ref{L: Forest Prob}. Lemma \ref{L: Ystar given Yk} and Lemma \ref%
{L: Pr Y star given Y1} imply that $\Pr \left[ Y^{\ast }\mid Y_{i}\right]
<\Pr \left[ Y^{\ast }\mid Y_{\pi \left( i\right) }\right] $ whereas, by
definition, $\Pr \left[ Y^{\ast }\mid N_{i}\right] =\Pr \left[ Y^{\ast }\mid
N_{\pi \left( i\right) }\right] =0$. Thus, $\pi \left( i\right) $'s
posterior beliefs about $Y^{\ast }$ are a mean-preserving spread of $i$'s
beliefs about $Y^{\ast }$ versus $\Omega \diagdown Y^{\ast }.$

Now from (i), $\mathcal{T\succeq }_{FO}\mathcal{T}^{\prime }$ if and only if 
$\mathcal{T}\succeq _{dif}\mathcal{T}^{\prime }$ and so there exists a
permutation $\pi $ such that for each $i,$ $d\left( i\right) \leq d\left(
\pi \left( i\right) \right) .$ Now the argument above shows that this is
equivalent to $\mathcal{T}\preceq _{SO}\mathcal{T}^{\prime }.\bigskip $
\end{proof}

The proposition establishes that there is a trade-off between the quality of
information about $G$ and the quality of information about $Y^{\ast }.$
While this trade-off by itself is insufficient to establish our irrelevance
result, it does offer some intuition why the irrelevance might hold. In
particular, the notion of common $p$-belief also employs third- and
higher-order information.

\section{Irrelevance result for games\label{S: Game}}

In this section we show that for many games of interest, the results of the
Section \ref{S: Main} concerning approximate common knowledge have \emph{%
exact} counterparts concerning equilibria of certain games. We first return
to the technology adoption game from the Introduction.

\subsection{Technology adoption game}

Recall that in the adoption game, each of $I$ agents can decide to adopt a
new technology (choose action $a_{i}=1$) or not ($a_{i}=0$). The cost of
adoption is $c$ per person. Adoption yields a gross payoff of $1$ if
everyone else adopts and the state is $g.$ Otherwise, the gross payoff is
zero.

Consider a forest $F=\mathcal{(}T^{1},T^{2},...,T^{R})$ and a seeding $%
s=\left( s^{1},s^{2},...,s^{R}\right) ,$ where $s^{r}$ is the unique seed of
tree $T^{r}.$ Again, denote by $\mathcal{T}$ the resulting (directed)
information tree. Let $\mathcal{E}\left( \mathcal{T}\right) $ be the set of
(possibly mixed) equilibria of the adoption game in which the information to
the agents comes via the network and seeding pair $\mathcal{T}=\left(
F,s\right) $.

The counterpart of Proposition \ref{P: small} is for any $\mathcal{T}$,

\begin{proposition}
\label{P: eqm small}Let agent $1$ be a seed of some tree. If $0<\varepsilon <%
\overline{\varepsilon },$ then the highest equilibrium probability that
everyone adopts is $\mathbb{P}_{\mathcal{T}}\left[ C_{\mathcal{T}}^{p}\left(
G\right) \right] $ where $p=c.$ Precisely, for any $\left( F,s\right) $ 
\[
\max_{\sigma \in \mathcal{E}\left( \mathcal{T}\right) }\mathbb{P}_{\mathcal{T%
}}\left[ \boldsymbol{a}=\mathbf{1\mid }\sigma \right] =\left\{ 
\begin{tabular}{cl}
$1$ & if $c\leq \Pr \left[ G\mid N_{1}\right] $ \\ 
$\Pr \left[ Y^{\ast }\right] $ & if $\Pr \left[ G\mid N_{1}\right] <c\leq
\Pr \left[ Y^{\ast }\mid Y_{1}\right] $ \\ 
$0$ & if $c>\Pr \left[ Y^{\ast }\mid Y_{1}\right] $%
\end{tabular}%
\right. 
\]
\end{proposition}

\paragraph{Proof}

\subparagraph{Case 1: $c\leq \Pr \left[ G\mid N_{1}\right] .$}

In this case, there is an equilibrium in which everyone adopts regardless of
whether he is informed or not. To see this, note that if everyone but $i$
always adopts, then the only uncertainty facing any agent is whether the
fundamental state is $g$ or $b.$ Lemma \ref{L: G given Nk} implies that for
all $i,$ $\Pr \left[ G\mid N_{i}\right] \geq \Pr \left[ G\mid N_{1}\right]
\geq c$, every agent is willing to adopt even if uninformed. Since everyone
adopts regardless of information, the probability that everyone adopts is $%
1. $

\subparagraph{Case 2: $\Pr \left[ G\mid N_{1}\right] <c\leq \Pr \left[
Y^{\ast }\mid Y_{1}\right] $.}

In this case, there is an equilibrium in which everyone adopts if and only
if informed. Moreover, there is no equilibrium in which an agent adopts with
positive probability when uninformed.

There are two steps to the argument.

\emph{Step 1}: $\Pr \left[ G\mid N_{1}\right] <c$ implies that no uninformed
agent adopts.

Consider the unique path from $0$ to $k$ and suppose (after renaming, if
necessary) that this path consists of agents $1,2,...,k$ such that the
direct predecessor of $j$ is $j-1.$

Then since $c>\Pr \left[ G\mid N_{1}\right] ,$ an uninformed agent $1$ does
not adopt even if everyone else adopts. In other words, adopting is
dominated for an uninformed agent $1.$

Now from Lemma \ref{L: Yk-1 given Nk}, $\Pr \left[ Y_{1}\mid N_{2}\right]
<\Pr \left[ G\mid N_{1}\right] $ which is less than $c.$ In other words,
adopting is iteratively dominated for an uninformed agent $2.$ Proceeding in
this way we see that for all agents $j$ along the path $1,2,...,k,$ adopting
is iteratively dominated for an uninformed agent $j.$ Since $k$ was
arbitrary, we have shown that adopting is iteratively dominated for all
agents.

\emph{Step 2}: $c\leq \Pr \left[ Y^{\ast }\mid Y_{1}\right] $ implies that
if all other agents adopt when informed, it is a best response for an
informed agent $i$ to do so as well.

To see why, suppose all agents but $i$ adopt when informed. Since $c\leq \Pr %
\left[ Y^{\ast }\mid Y_{1}\right] ,$ Lemma \ref{L: Ystar given Yk} implies
that for all $i,$ $c<\Pr \left[ Y^{\ast }\mid Y_{i}\right] ,$ and so it is a
best response for agent $i$ to adopt as well. Thus, there exists an
equilibrium in which everyone adopts if and only if informed.

In this case, the probability that everyone adopts is just the probability
of $Y^{\ast }.$ Because of Step 1, no equilibrium can involve adopting with
positive probability when uninformed. Thus, the equilibrium in which every
informed agent adopts gives the highest probability of adoption.

\subparagraph{Case 3: $c>\Pr \left[ Y^{\ast }\mid Y_{1}\right] .$}

In this case, the unique equilibrium is one in which no one ever adopts.

To see why, note that the cost is so high that even if everyone else adopts
only if informed, it is a best response for an informed agent $1$ to not
adopt. Thus, agent $1$ will never adopt.

This implies that no agent will ever adopt. Thus, the only equilibrium is
one in which no one ever adopts. Of course, the probability of adoption is
then $0.\bigskip $

This completes the proof. $\blacksquare \bigskip $

Proposition \ref{P: eqm small} can also be derived as a consequence of
Proposition \ref{P: small} using some results from Kajii and Morris (1997)
and Oyama and Takahashi (2020). We have chosen to provide a self-contained
proof of the former as it emphasizes the parallel nature of the proofs of
the two propositions. A counterpart of Proposition \ref{P: large} can
similarly be derived.

\subparagraph{Payoffs}

What about agents' payoffs in the technology adoption game? It is easy to
see that while the maximum equilibrium probability of adoption is
independent of the information tree $\mathcal{T},$ agents' payoffs do depend
on $\mathcal{T}.$ This is easily verified in the three-agent example in the
Introduction. Suppose $\varepsilon $ is small and $c$ is in the intermediate
range. When $1$ is the seed, the expected payoffs are $u_{i}=\rho \left(
1-\varepsilon \right) ^{3}-\left( 1-\varepsilon \right) ^{i}c.$ Note that $%
u_{1}<u_{2}<u_{3}$ and so the agent with the worst information about $\theta 
$ has the highest expected payoff. With broadcasting, all three agents have
the same payoff $\rho \left( 1-\varepsilon \right) ^{3}-\left( 1-\varepsilon
\right) c.$

It is also the case that the equilibria identified in Proposition \ref{P:
eqm small} not only maximize the probability that everyone adopts, but also
Pareto dominate all other equilibria in terms of payoffs.

\begin{corollary}
\label{C: Pareto}In each case, the equilibrium from Proposition \ref{P: eqm
small} that maximizes the probability that everyone adopts also Pareto
dominates all other equilibria.
\end{corollary}

\begin{proof}
To see that in each case the identified equilibrium Pareto dominates all
other equilibria, first let $\left( \alpha _{i},\beta _{i}\right) $ be the
randomized strategy of $i$ in which when informed, he adopts with
probability $\alpha _{i}$ and when uninformed, adopts with probability $%
\beta _{i}.$ Define $u_{i}\left( \alpha _{i},\beta _{ii},\boldsymbol{\alpha }%
_{-i},\boldsymbol{\beta }_{-i}\right) $ to be $i$'s expected payoff when he
plays strategy $\left( \alpha _{i},\beta _{i}\right) $ and the others play $%
\boldsymbol{\alpha }_{-i}=\left( \alpha _{j}\right) _{j\neq i}$ and $%
\boldsymbol{\beta }_{-i}=\left( \beta _{j}\right) _{j\neq i}.$

In Case 1, the fact that adopting exerts positive externalities implies that%
\[
u_{i}\left( \alpha _{i},\beta _{i},\boldsymbol{\alpha }_{-i},\boldsymbol{%
\beta }_{-i}\right) \leq u_{i}\left( \alpha _{i},\beta _{i},\boldsymbol{1}%
_{-i},\boldsymbol{1}_{-i}\right) 
\]%
where $\boldsymbol{\alpha }_{-i}=\boldsymbol{1}_{-i}$ means that all $j\neq
i $ play $\alpha _{j}=1$ and the similarly for $\boldsymbol{\beta }_{-i}.$
But since $\left( \alpha _{i},\beta _{i}\right) =\left( 1,1\right) $ is a
best-response to $\left( \boldsymbol{1}_{-i},\boldsymbol{1}_{-i}\right) ,$
we have that%
\[
u_{i}\left( \alpha _{i},\beta _{i},\boldsymbol{\alpha }_{-i},\boldsymbol{%
\beta }_{-i}\right) \leq u_{i}\left( 1,1,\boldsymbol{1}_{-i},\boldsymbol{1}%
_{-i}\right) 
\]

In Case 2, suppose that there is an equilibrium $\left( \alpha _{j},\beta
_{j}\right) $ for $j\in \mathcal{I}.$ Since adopting is iteratively
dominated when uninformed, it must be that in any equilibrium $\beta _{j}=0$
for all $j.$ Now again since adopting exerts positive externalities,%
\[
u_{i}\left( \alpha _{i},0,\boldsymbol{\alpha }_{-i},\boldsymbol{0}%
_{-i}\right) \leq u_{i}\left( \alpha _{i},0,\boldsymbol{1}_{-i},\boldsymbol{0%
}_{-i}\right) 
\]%
where $\boldsymbol{\beta }_{-i}=\mathbf{0}_{-i}$ means that all $j\neq i$
play $\beta _{j}=0.$ Since $\left( \alpha _{i},\beta _{i}\right) =\left(
1,0\right) $ is a best-response to $\left( \boldsymbol{1}_{-i},\boldsymbol{0}%
_{-i}\right) ,$%
\[
u_{i}\left( \alpha _{i},0,\boldsymbol{\alpha }_{-i},\boldsymbol{0}%
_{-i}\right) \leq u_{i}\left( 1,0,\boldsymbol{1}_{-i},\boldsymbol{0}%
_{-i}\right) 
\]

In Case 3 the equilibrium is unique.\bigskip
\end{proof}

In each case of Proposition \ref{P: small}, the equilibrium that maximizes
the probability that everyone adopts is symmetric---all agents play the same
strategy. But not all equilibria are symmetric. Consider a situation in
which $1$ and $2$ are connected and $3$ is isolated (a trivial tree).
Suppose $1$ and $3$ are the seeds for each of the two trees. Then if $\rho >%
\frac{1}{2}$ and $\varepsilon $ is small enough, for an open set of $c$'s
there is an asymmetric equilibrium in which $1$ and $2$ adopt if and only if
informed whereas $3$ always adopts. Of course, the corollary above implies
that this asymmetric equilibrium is Pareto dominated by one in which
everyone always adopts.

In the technology adoption game, the players exert positive externalities on
each other. The reader may then wonder if the irrelevance result relies on
this feature. This is not the case. In the game below, the irrelevance
result holds even though there are \emph{negative} externalities.

\subsection{Negative externalities}

A group of $I$ agents try to overthrow a repressive regime by protesting. If
all $I$ agents protest and the state is $g$, then the regime falls.
Otherwise, it survives. But protests can turn destructive---by damaging
public buildings, burning buses etc.. The greater the number of protesters,
the greater the damage. Let $a_{i}=1$ denote that $i$ protests and $a_{i}=0$
that he does not.

Suppose that in state $g$, the payoff functions are as follows:%
\begin{eqnarray*}
u_{i}^{g}\left( 1,\boldsymbol{a}_{-i}\right) &=&\left\{ 
\begin{array}{cc}
R-c-Id & \text{if }\sum_{j}a_{j}=I \\ 
-c-\left( \sum_{j}a_{j}\right) d & \text{if }\sum_{j}a_{j}<I%
\end{array}%
\right. \\
u_{i}^{g}\left( 0,\boldsymbol{a}_{-i}\right) &=&-\left(
\sum\nolimits_{j\neq i}a_{j}\right) d
\end{eqnarray*}%
In state $b$, they are%
\begin{eqnarray*}
u_{i}^{b}\left( 1,\boldsymbol{a}_{-i}\right) &=&-c-\left(
\sum\nolimits_{j}a_{j}\right) d \\
u_{i}^{b}\left( 0,\boldsymbol{a}_{-i}\right) &=&-\left(
\sum\nolimits_{j\neq i}a_{j}\right) d
\end{eqnarray*}

Here each player incurs a personal cost $c>0$ of protesting (either an
opportunity cost or one resulting from the risk of arrest or injury). Each
additional protester also causes an additional $d>0$ damage to public
property and this damage reduces the payoff to everybody. If the regime is
overthrown, then each player gets a reward $R>c+d$, the private plus public
marginal cost of protesting. This last condition ensures that in state $g,$
it is an equilibrium for everyone to protest.

In the protest game, players exert a \emph{negative} externality on each
other because of the destruction of public property.

Then it can be verified that the counterpart of Propositions \ref{P: eqm
small} holds for the protest game as well. Simply replace "adopting" with
"protesting" and the parameter $c$ with the parameter $\left( c+d\right) /R.$

It should be noted, however, that Corollary \ref{C: Pareto} is special to
the technology adoption game that exhibits positive externalities and does
not extend to the destructive protest game with negative externalities.

\subsection{Potential games}

What is the general class of games for which the irrelevance result holds?

Consider an $I$-player \emph{symmetric} game in which each $i$ chooses an
action $a_{i}\in A=\left\{ 0,1\right\} $. Let $\boldsymbol{a}=\left(
a_{i}\right) _{i\in \mathcal{I}}$ denote the vector of actions of all the
players and let $\boldsymbol{a}_{-i}=\left( a_{j}\right) _{j\neq i}$ denote
the vector of actions of all the players except $i.$

There are two possible states of nature $\theta \in \left\{ g,b\right\} $
and the payoff functions in state $\theta $, $u_{i}^{\theta
}:A^{I}\rightarrow R$ are such that $u_{i}^{\theta }\left( a_{i},\boldsymbol{%
a}_{-i}\right) $ depends only on $i$'s own action $a_{i}$ and on the number
of other players who play $a_{j}=1,$ $\sum_{j\neq i}a_{j}.$

As defined by Monderer and Shapley (1996), a game with payoffs $%
u_{i}^{\theta }:A^{I}\rightarrow R$ is a \emph{potential} game if there is a 
\emph{potential function} $v^{\theta }:A^{I}\rightarrow R$ such that for all 
$i,$ $a_{i}$, $a_{i}^{\prime }$ and $a_{-i}$,%
\[
u_{i}^{\theta }\left( a_{i},\boldsymbol{a}_{-i}\right) -u_{i}^{\theta
}\left( a_{i}^{\prime },\boldsymbol{a}_{-i}\right) =v^{\theta }\left( a_{i},%
\boldsymbol{a}_{-i}\right) -v^{\theta }\left( a_{i}^{\prime },\boldsymbol{a}%
_{-i}\right) 
\]%
In other words, for each $\theta $, the game is best-response equivalent to
a game with common interests.

It is easily verified that \emph{every} binary-action symmetric game is a
potential game.

Consider games with potentials of the form%
\begin{eqnarray}
v^{g}\left( \boldsymbol{a}\right) &=&\left\{ 
\begin{array}{cc}
w & \text{if }\sum_{j}a_{j}=I \\ 
-\left( \sum_{j}a_{j}\right) \gamma & \text{if }\sum_{j}a_{j}<I%
\end{array}%
\right.  \nonumber \\
v^{b}\left( \boldsymbol{a}\right) &=&-\left( \sum\nolimits_{j}a_{j}\right)
\gamma  \label{potential}
\end{eqnarray}%
where $w$ and $\gamma >0$ are parameters that satisfy $w>-\left( I-1\right)
\gamma $. Let $\mathcal{P}$ denote the class of potential games of the form
in (\ref{potential}).

The requirement that $w>-\left( I-1\right) \gamma $ guarantees that in the
game with common payoffs $v^{g},$ it is a strict Nash equilibrium for
everyone to choose $a_{i}=1.$

It is easy to verify that the technology adoption game from the previous
subsection is a potential game from the class $\mathcal{P}$ where $\gamma =c$
and $w=1-I\gamma ,$ so that the last requirement that $w>-\left( I-1\right)
\gamma $ reduces to $1>c.$

Much along the lines of Proposition \ref{P: eqm small}, it can be shown that
if $\varepsilon <\overline{\varepsilon },$ then for \emph{any} potential
game in the class $\mathcal{P}$ the irrelevance result applies.

\begin{proposition}
Let agent $1$ be a seed of some tree. If $0<\varepsilon <\overline{%
\varepsilon },$ then the highest equilibrium probability that everyone plays 
$a_{i}=1$ is $\mathbb{P}_{\mathcal{T}}\left[ C_{\mathcal{T}}^{p}\left(
G\right) \right] $ where $p=\frac{\gamma }{w+\gamma I}.$ Precisely, for any $%
\mathcal{T}=\left( F,s\right) $ 
\[
\max_{\sigma \in \mathcal{E}\left( \mathcal{T}\right) }\mathbb{P}_{\mathcal{T%
}}\left[ \boldsymbol{a}=\mathbf{1\mid }\sigma \right] =\left\{ 
\begin{tabular}{cl}
$1$ & if $\frac{\gamma }{w+\gamma I}\leq \Pr \left[ G\mid N_{1}\right] $ \\ 
$\Pr \left[ Y^{\ast }\right] $ & if $\Pr \left[ G\mid N_{1}\right] <\frac{%
\gamma }{w+\gamma I}\leq \Pr \left[ Y^{\ast }\mid Y_{1}\right] $ \\ 
$0$ & if $\frac{\gamma }{w+\gamma I}>\Pr \left[ Y^{\ast }\mid Y_{1}\right] $%
\end{tabular}%
\right. 
\]
\end{proposition}

We omit a proof of this result as it mimics the proof of Proposition \ref{P:
eqm small}.

\section{Other networks and seedings\label{S: Other}}

In this section we explore some limits to our results. Specifically, we show
that the irrelevance result does not hold once the networks have cycles.
Also, it is sensitive to the assumption that each tree has a single
seed---multiple or random seedings can make the structure relevant.

In this section we state all of our findings using the technology adoption
game. These can all be restated in terms of common $p$-beliefs as well.

\subsection{Cycles\label{S: Nontree}}

Our result that the set $C^{p}\left( G\right) $ of states in which $G$ is
common $p$-believed rests crucially on the assumption the underlying network
is acyclic---a forest. Once there are cycles, the conclusion of our main
result does not hold.

Consider, for example, a situation with $4$ agents in the cyclic network
depicted in Figure \ref{F: Cycle}.%
%
%
%
%
%
%
%
%
%
%
%
%
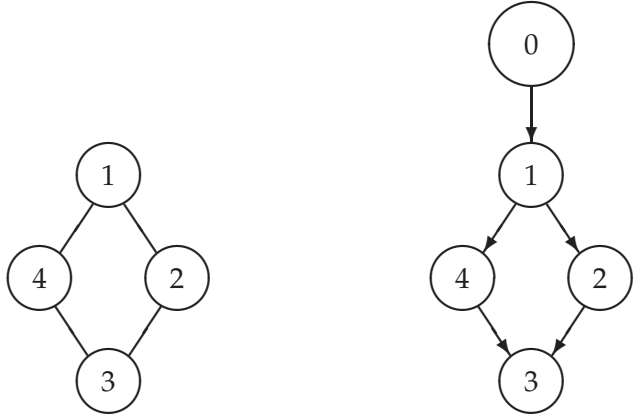
\begin{figure}[t]
\unitlength 2pt
\begin{picture}(160,90)(-12,-5)
\linethickness{0.5pt}
\small
\thicklines

\put(57,50){\circle {12}}
\put(57,50){\makebox(0,0)[cc]{$1$}}
\put(54,44.5){\line(-2,-3){6}}
\put(60,44.5){\line(2,-3){6}}

\put(44,30.5){\circle {12}}
\put(44,30.5){\makebox(0,0)[cc]{$4$}}
\put(47,25){\line(2,-3){6}}

\put(70,30.5){\circle {12}}
\put(70,30.5){\makebox(0,0)[cc]{$2$}}
\put(67,25){\line(-2,-3){6}}

\put(57,11){\circle {12}}
\put(57,11){\makebox(0,0)[cc]{$3$}}

\put(137,74.8){\circle {16}}
\put(137,74.8){\makebox(0,0)[cc]{$0$}}
\put(137,66.8){\vector(0,-1){10.8}}

\put(137,50){\circle {12}}
\put(137,50){\makebox(0,0)[cc]{$1$}}
\put(134,44.5){\vector(-2,-3){6}}
\put(140,44.5){\vector(2,-3){6}}

\put(124,30.5){\circle {12}}
\put(124,30.5){\makebox(0,0)[cc]{$4$}}
\put(127,25){\vector(2,-3){6}}

\put(150,30.5){\circle {12}}
\put(150,30.5){\makebox(0,0)[cc]{$2$}}
\put(147,25){\vector(-2,-3){6}}

\put(137,11){\circle {12}}
\put(137,11){\makebox(0,0)[cc]{$3$}}

\end{picture}
\caption{Cycle Network}
\label{F: Cycle}
\end{figure}
Suppose that the planner seeds the network by sending a message to $1.$ The
resulting information network, depicted on the right-hand side of the
figure, is now not a directed tree. Here, $3$ can receive messages from two
sources---either from $2$ or from $4.$ And of course, what $3$ believes
about whether agents $2$ and $4$ are informed depends on whom she hears
from. If $3$ receives a message only from $2$, then she knows $1$ and $2$
are informed but is unsure about $4$. If she hears from both $2$ and $4$,
then she knows that everyone else is informed.

So let $Y_{3}^{2}$ denote the event that agent $3$ heard a message \emph{only%
} directly from $2$ and similarly, let $Y_{3}^{4}$ denote the event that
agent $3$ heard only from $4.$ Finally, let $Y_{3}^{2\wedge 4}$ denote the
event that she heard from both $2$ and $4.$ As before, let $N_{3}$ denote
the event that $3$ did not hear from either source. The events $Y_{3}^{2},$ $%
Y_{3}^{4},$ $Y_{3}^{2\wedge 4}$ and $N_{3}$ are mutually exclusive and
exhaustive.

Consider the adoption game with four players. We will compare equilibrium
outcomes in the cycle network with those in the line network with four
agents.

\begin{claim}
Consider the four-agent technology adoption game. When $\varepsilon $ is
small, for an open set of costs $c$, there is an equilibrium with a line
network seeded via $1$ in which the probability that everyone adopts is
greater than that from any equilibrium with the cycle network again seeded
via $1.$
\end{claim}

The proof of the claim is in Appendix \ref{A: Nontree}.

\subsection{Multiple seeds\label{S: Mult}}

We have assumed that the planner sends information to a \emph{single} agent
in each tree of the forest network---that is, each tree has a single seed.
Intuition suggests that it is better to send information to many agents at
once---that is, to create multiple seeds. This will surely help reduce
first-order uncertainty about the state. Here we show that the effect of
multiple seeds on reducing higher-order uncertainty is ambiguous and in some
circumstances a single seed is "better" than multiple seeds. In the
technology adoption game from the introduction, there is a range of costs $c$
for which agents' welfare is higher with a single seed than with multiple
seeds.%
%
%
%
%
%
%
%
%
%
%
%
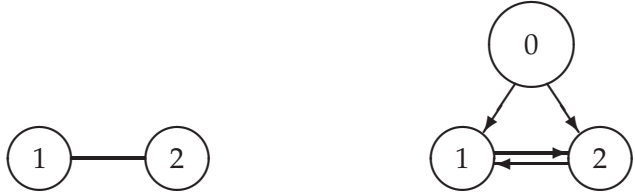
\begin{figure}[t]
\unitlength 2pt
\begin{picture}(160,60)(-12,-5)
\linethickness{0.5pt}
\small
\thicklines

\put(44,10.5){\circle {12}}
\put(44,10.5){\makebox(0,0)[cc]{$1$}}
\put(50,10.5){\line(1,0){14}}

\put(70,10.5){\circle {12}}
\put(70,10.5){\makebox(0,0)[cc]{$2$}}

\put(137,32){\circle {16}}
\put(137,32){\makebox(0,0)[cc]{$0$}}
\put(134,24.5){\vector(-2,-3){6}}
\put(140,24.5){\vector(2,-3){6}}

\put(124,10.5){\circle {12}}
\put(124,10.5){\makebox(0,0)[cc]{$1$}}
\put(130,11.5){\vector(1,0){14}}

\put(150,10.5){\circle {12}}
\put(150,10.5){\makebox(0,0)[cc]{$2$}}
\put(144,9.5){\vector(-1,0){14}}

\end{picture}
\caption{Mutiple Seeds}
\label{F: Multiseed}
\end{figure}%

Consider a simple network with two connected agents (as in left-hand panel
of Figure \ref{F: Multiseed}). Now suppose the planner seeds \emph{both} $1$
and $2.$ This results in the information network depicted in the right-hand
panel. Here if $1$ gets a message, she passes it along to $2$ and vice
versa. Thus each agent has two sources of information---directly from the
planner or indirectly from the other agent. And of course, what $1$ believes
about whether $2$ is also informed depends on the channel by which she
received a message. If $1$ received a message from $2$, then she knows that $%
2$ is also informed. But if $1$ received a message directly from $0$, and
not from $2,$ then she is unsure about whether $2$ is informed.

Suppose that the agents play the two-player version of the technology
adoption game from Section \ref{S: Game}.

For the connected two-player network in Figure \ref{F: Multiseed}, we have

\begin{claim}
\label{c: multgame}Consider the two-agent technology adoption game. When $%
\varepsilon $ is small, for an open set of costs $c$, there is an
equilibrium with a single seed in which the probability that everyone adopts
is greater than that from any equilibrium with two seeds.
\end{claim}

The proof of this claim is in Appendix \ref{A: Mult}.

\subsection{Random seeds}

We have assumed that the planner chooses a single seed in each tree and the
identity of the seed is known to all the agents. Does the irrelevance result
still hold if the planner choose a single seed, but at \emph{random}? The
answer is no as the following example demonstrates.

Suppose there are only two agents, $1$ and $2$ and the network is connected.
With probability $\frac{1}{2}$ the planner chooses $1$ as the seed and with
probability $\frac{1}{2}$ chooses $2$ as the seed.

For the connected two-player network, we have

\begin{claim}
Consider the two-agent technology adoption game. When $\varepsilon $ is
small, for an open set of costs $c$, there is an equilibrium with a random
seeding in which the probability that everyone adopts is greater than that
from any equilibrium with only one seed.
\end{claim}

The proof of this claim is in Appendix \ref{A: Random}.

\section{Conclusion}

People are interconnected in many ways. The same person may be part of a
professional network, a family network or a leisure network and so have many
sources of information. In such situations, the overall network may not have
a tree-like structure and so our irrelevance result will not apply.
Nevertheless, the point that there is a trade-off between disseminating
information quickly and making the information commonly known can be applied
more generally. When the goal of the policymaker is to engineer coordinated
behavior, the latter is more important.

\appendix%

\section{Appendix: Agents' beliefs\label{A: Beliefs}}

This appendix derives agents' beliefs of different events used to prove the
main result.

\subsection{Beliefs along a path}

Here we derive three results that compare the beliefs of agents who lie
along the same path originating with the planner (agent $0$). So let $1$ be
a seed and let $K$ be a terminal node (a leaf). Suppose, after renaming,
that the unique path from $0$ to $K$ consists of agents $k=1,2,..,K$ such
that the direct predecessor of $k$ is $k-1.$

The first lemma simply says that the further an uninformed agent is from the
seed, the more optimistic he is that the fundamental $\theta =g.$ This is
because an uninformed agent further from the seed ascribes a higher
probability to the event that the fundamental is $g$ and the message got
lost somewhere along the way, than someone closer to the seed.

\begin{lemma}
\label{L: G given Nk}The sequence%
\[
\Pr \left[ G\mid N_{k}\right] 
\]%
is increasing in $k.$
\end{lemma}

\begin{proof}
Note that%
\begin{equation}
\Pr \left[ G\mid N_{k}\right] =\frac{\rho \left( 1-\left( 1-\varepsilon
\right) ^{k}\right) }{1-\rho \left( 1-\varepsilon \right) ^{k}}
\label{Pr G given Nk}
\end{equation}%
since $\Pr \left[ Y_{k}\right] =\rho \left( 1-\varepsilon \right) ^{k}$ and
so $\Pr \left[ N_{k}\right] =1-\rho \left( 1-\varepsilon \right) ^{k}.$

The result then follows immediately.\bigskip
\end{proof}

The second lemma says that the further an uninformed agent is from the seed,
the more pessimistic he is about the event that all his predecessors are
informed. The intuition is that the further the agent is along the path, the
greater the chance that the message was lost somewhere prior to reaching his
immediate predecessor.

In what follows, it will be convenient to write%
\begin{equation}
Y_{0}=G  \label{def Y0}
\end{equation}

\begin{lemma}
\label{L: Yk-1 given Nk}The sequence%
\[
\Pr \left[ Y_{k-1}\mid N_{k}\right] 
\]%
is decreasing in $k$ for $k\geq 1.$
\end{lemma}

\begin{proof}
Note that since $Y_{0}=G,$ for any $k\geq 1$,%
\[
\Pr \left[ Y_{k-1}\mid N_{k}\right] =\frac{\rho \left( 1-\varepsilon \right)
^{k-1}\varepsilon }{1-\rho \left( 1-\varepsilon \right) ^{k}} 
\]%
The numerator is the probability of the joint event $Y_{k-1}\cap N_{k}$
which occurs if only if $k-1$ receives the message (probability $\rho \left(
1-\varepsilon \right) ^{k-1}$) and $k$ does not (probability $\varepsilon $%
). The denominator is the probability of $N_{k}$ which is just $1-\Pr \left[
Y_{k}\right] $. Now $Y_{k}$ occurs if and only if $k$ gets the message
(probability $\rho \left( 1-\varepsilon \right) ^{k}$).

It is now easy to verify that%
\[
\Pr \left[ Y_{k-1}\mid N_{k}\right] >\Pr \left[ Y_{k}\mid N_{k+1}\right] 
\]%
\bigskip
\end{proof}

The third lemma is also rather intuitive. It says that informed agents who
are further along the path are increasingly optimistic that all agents,
whether or not they are on the path, are informed as well.

\begin{lemma}
\label{L: Ystar given Yk}The sequence%
\[
\Pr \left[ Y^{\ast }\mid Y_{k}\right] 
\]%
is increasing in $k.$
\end{lemma}

\begin{proof}
Since for all $k,$%
\[
\Pr \left[ Y^{\ast }\mid Y_{k}\right] \times \Pr \left[ Y_{k}\right] =\Pr %
\left[ Y^{\ast }\right] 
\]%
we have%
\[
\frac{\Pr \left[ Y^{\ast }\mid Y_{k-1}\right] }{\Pr \left[ Y^{\ast }\mid
Y_{k}\right] }=\frac{\Pr \left[ Y_{k}\right] }{\Pr \left[ Y_{k-1}\right] } 
\]%
And since $Y_{k}\subset Y_{k-1}$ and $k-1$ is the unique direct predecessor
of $k,$%
\begin{eqnarray*}
\Pr \left[ Y_{k}\right] &=&\Pr \left[ Y_{k}\mid Y_{k-1}\right] \times \Pr %
\left[ Y_{k-1}\right] \\
&=&\left( 1-\varepsilon \right) \times \Pr \left[ Y_{k-1}\right]
\end{eqnarray*}%
and so%
\begin{equation}
\frac{\Pr \left[ Y^{\ast }\mid Y_{k-1}\right] }{\Pr \left[ Y^{\ast }\mid
Y_{k}\right] }=1-\varepsilon  \label{Pr ratio}
\end{equation}
\end{proof}

\subsection{Probability of $Y^{\ast }$}

Suppose the $I$ agents are completely disconnected so that each agent is a
"tree" with one node (as in Example $1$ case B). Now the only way the
information can get to all the agents is if every agent is a seed---that is,
the planner "broadcasts" the message. Since each message is lost
independently with probability $\varepsilon ,$ the probability that the
information reaches all the agents is simply $\left( 1-\varepsilon \right)
^{I}.$

In an arbitrary tree (or more generally, a forest), if a message from $i$ to 
$j$ is lost so that $j$ is uninformed, then this means that all agents in
the sub-tree with $j$ as the root are also uninformed. Thus, unlike in the
case of a broadcast, whether or not $i$ and $j$ are informed are correlated.
The next result shows that despite this, no matter what the structure of the
forest is, the probability that all agents are informed is the same as when
there is a broadcast.

\begin{lemma}
\label{L: Forest Prob}For any forest with $I$ nodes,%
\[
\Pr \left[ Y^{\ast }\mid G\right] =\left( 1-\varepsilon \right) ^{I} 
\]
\end{lemma}

\begin{proof}
The proof is by induction on $I.$

For $I=1,$ clearly the probability $\Pr \left[ Y^{\ast }\mid G\right]
=1-\varepsilon .$

Now suppose that for any forest with $I-1$ agents%
\[
\Pr \left[ \cap _{i=1}^{I-1}Y_{i}\mid G\right] =\left( 1-\varepsilon \right)
^{I-1} 
\]%
In the forest with $I$ agents, let $I$ be a leaf (a terminal node) of some
tree in the forest and let the unique direct predecessor of $I$ be $I-1.$
Then since%
\[
\Pr \left[ \cap _{i=1}^{I}Y_{i}\mid G\right] =\Pr \left[ \cap
_{i=1}^{I-1}Y_{i}\mid G\right] \times \Pr \left[ Y_{I}\mid \cap
_{i=1}^{I-1}Y_{i},G\right] 
\]%
and $\Pr \left[ Y_{I}\mid \cap _{i=1}^{I-1}Y_{i},G\right] =1-\varepsilon ,$
the claim is established.\bigskip
\end{proof}

A simple consequence of the previous result is

\begin{lemma}
\label{L: Pr Y star given Y1}For any forest, 
\[
\Pr \left[ Y^{\ast }\mid Y_{1}\right] =\left( 1-\varepsilon \right) ^{I-1} 
\]
\end{lemma}

\begin{proof}
The proof just mimics the proof of Lemma \ref{L: Forest Prob}.\bigskip
\end{proof}

Combined with the fact that for successive agents along the path from $1$ to 
$K,$ 
\[
\Pr \left[ Y^{\ast }\mid Y_{k-1}\right] =\left( 1-\varepsilon \right) \times
\Pr \left[ Y^{\ast }\mid Y_{k}\right] 
\]
(see (\ref{Pr ratio})), Lemma \ref{L: Pr Y star given Y1} implies that for
all $k,$ 
\begin{equation}
\Pr \left[ Y^{\ast }\mid Y_{k}\right] =\left( 1-\varepsilon \right) ^{I-k}
\label{power I-k}
\end{equation}

\section{Appendix: Other networks and seedings}

\subsection{Cycles\label{A: Nontree}}

For the cyclical network of Section \ref{S: Nontree} we have

\begin{claim}
Suppose $\frac{1-\varepsilon }{2-\varepsilon }<c<\left( 1-\varepsilon
\right) ^{4}$ and $\frac{\rho \varepsilon }{\rho \varepsilon +1-\rho }<c.$
Then with the cycle network, there is an equilibrium in which $i=1,2,4$
adopt if and only if informed and $3$ adopts if only if he is informed via
both $2$ and $4.$
\end{claim}

\begin{proof}
First, since $c>\Pr \left[ G\mid N_{1}\right] =\frac{\rho \varepsilon }{\rho
\varepsilon +1-\rho }$ it is iteratively dominated for any uninformed agent
to adopt. Clearly, it is dominated for $N_{1}$ to adopt. Given this, Lemma %
\ref{L: Yk-1 given Nk}, it is iteratively dominated for $N_{2}$ and $N_{4}$
to adopt.

Second, given that $N_{1},N_{2}$ and $N_{4}$ do not adopt, it is iteratively
dominated for $Y_{3}^{2}$ to adopt. This is because if $3$ got a message
from $2,$ she is sure that both $1$ and $2$ are informed. The only
uncertainty she faces concerns agent $4$ and it is easily verified that $\Pr %
\left[ Y_{4}\mid Y_{3}^{2}\right] =\left( 1-\varepsilon \right) /\left(
2-\varepsilon \right) .$ Since this is smaller than $c,$ $Y_{3}^{2}$ should
not adopt. By interchanging the roles of $2$ and $4,$ we infer that $%
Y_{3}^{4}$ should not adopt either.

Third, given the specified strategies, it is a best-response for $Y_{1}$ to
adopt since the probability that all others adopt is $\Pr \left[ Y_{2}\cap
Y_{4}\cap Y_{3}^{2\wedge 4}\mid Y_{1}\right] =\left( 1-\varepsilon \right)
^{4}.$ To see this note that 
\begin{eqnarray}
\Pr \left[ Y_{1}\cap Y_{2}\cap Y_{4}\cap Y_{3}^{2\wedge 4}\right] &=&\Pr %
\left[ Y_{1}\right] \Pr \left[ Y_{2}\cap Y_{4}\mid Y_{1}\right] \Pr \left[
Y_{3}^{2\wedge 4}\mid Y_{2}\cap Y_{4}\right]  \label{Pr cycle} \\
&=&\Pr \left[ Y_{1}\right] \times \left( 1-\varepsilon \right) ^{2}\times
\left( 1-\varepsilon \right) ^{2}  \nonumber
\end{eqnarray}%
Since $c<\left( 1-\varepsilon \right) ^{4},$ it is a best-response for $%
Y_{1} $ to adopt.

Given the strategies, $Y_{2},Y_{4}$ and $Y_{3}^{2\wedge 4}$ are all more
optimistic than $Y_{1}$ about the event that everyone will adopt. So they
too will adopt.\bigskip
\end{proof}

From (\ref{Pr cycle}), the probability that everyone will adopt in the
equilibrium described in the claim above is just $\left( 1-\varepsilon
\right) ^{5}$ and this is the highest achievable since $N_{1},$ $N_{2},$ $%
N_{4},$ $Y_{3}^{2}$ and $Y_{3}^{4}$ do not adopt in any equilibrium.

In the line network (or any tree), the corresponding probability is $\left(
1-\varepsilon \right) ^{4}$.

\subsection{Multiple seeds\label{A: Mult}}

When both $1$ and $2$ are seeds, let $Y_{i}^{0}$ denote the event that $i$
heard only directly from the planner, $Y_{i}^{j}$ the event that $i$ heard
only from agent $j=3-i,$ and $Y_{i}^{0\wedge j}$ the event that $i$ heard
from both the planner and $j\neq i.$ Finally, let $N_{i}$ be the event that $%
i$ hears from neither. Let $Y_{i}$ denote the event that $i$ heard from
either source, that is, $Y_{i}=Y_{i}^{0}\cup Y_{i}^{j}\cup Y_{i}^{0\wedge
j}. $

In effect, there are now four types of each agent and thus the states of the
world are more complicated since they specify not only whether or not $i$ is
informed but the source of her information. Let $\Omega _{M}$ be the states
of the world for the example when there are multiple (two) seeds.

With multiple seeds, 
\begin{eqnarray}
\mathbb{P}_{M}\left[ G\mid N_{i}\right] &=&\frac{\rho \varepsilon ^{2}+\rho
\varepsilon \left( 1-\varepsilon \right) \varepsilon }{\rho \varepsilon
^{2}+\rho \varepsilon ^{2}\left( 1-\varepsilon \right) +1-\rho }  \nonumber
\\
&=&\frac{\rho \varepsilon ^{2}\left( 2-\varepsilon \right) }{\rho
\varepsilon ^{2}\left( 2-\varepsilon \right) +1-\rho }  \label{PrM G N}
\end{eqnarray}%
where the probabilities $\mathbb{P}_{M}$ are now determined in the network
with multiple seeds. The numerator is the probability that if $\theta =g,$
neither hears from $0$ (probability $\varepsilon ^{2}$) plus the probability
that $i$ does not hear from $0$ (probability $\varepsilon $) but $j$ does
(probability $1-\varepsilon $) and then $j$'s message is lost (probability $%
\varepsilon $). The denominator is just the probability that $i$ hears from
neither source.

First, note that if $1$ hears directly only from $0,$ then her belief that $%
2 $ is informed is%
\begin{eqnarray*}
\mathbb{P}_{M}\left[ Y_{2}\mid Y_{1}^{0}\right] &=&\left( 1-\varepsilon
\right) +\varepsilon \left( 1-\varepsilon \right) \\
&=&1-\varepsilon ^{2}
\end{eqnarray*}%
and so from (\ref{PrM G N}) it follows that for $\varepsilon $ small enough,%
\[
\mathbb{P}_{M}\left[ G\mid N_{1}\right] <\mathbb{P}_{M}\left[ Y_{2}\mid
Y_{1}^{0}\right] 
\]

\begin{claim}
With multiple seeds, if $\mathbb{P}_{M}\left[ G\mid N_{1}\right] <c\leq 
\mathbb{P}_{M}\left[ Y_{2}\mid Y_{1}^{0}\right] ,$ then there is an
equilibrium in which agents adopt if and only if they get a message from
either source. This equilibrium Pareto dominates all other equilibria.
\end{claim}

\begin{proof}
If $i$ does not get a message, then his belief about the event $G$ is
smaller than the cost and so it is dominant to not adopt.

Suppose $2$ adopts whenever she is informed. Now since, $c\leq \Pr \left[
Y_{2}\mid Y_{1}^{0}\right] ,$ if $1$ gets a signal only from $0,$ he will
adopt. And if $1$ hears from $2$, then he knows that $2$ is also informed
and so will also adopt.\bigskip
\end{proof}

For the range of costs in the claim above, with multiple seeds, the
resulting equilibrium payoff can be calculated as follows.

The probability that both adopt in the event $G$ is%
\begin{eqnarray*}
\mathbb{P}_{M}\left[ Y_{1}\cap Y_{2}\mid G\right] &=&2\left( 1-\varepsilon
\right) ^{2}\varepsilon +\left( 1-\varepsilon \right) ^{2} \\
&=&\left( 1-\varepsilon \right) ^{2}\left( 1+2\varepsilon \right)
\end{eqnarray*}%
The probability that $1$ adopts and $2$ doesn't is%
\[
\mathbb{P}_{M}\left[ Y_{1}\cap N_{2}\mid G\right] =\left( 1-\varepsilon
\right) \varepsilon ^{2} 
\]%
Since adoption occurs only in the event $G,$ the ex ante equilibrium payoff
of either agent when both are seeds is%
\begin{equation}
\pi _{M}=\rho \left( 1-\varepsilon \right) ^{2}\left( 1+2\varepsilon \right)
\left( 1-c\right) +\rho \left( 1-\varepsilon \right) \varepsilon ^{2}\left(
-c\right)  \label{mult payoff}
\end{equation}

\subparagraph{Single seed}

If $1$ is the only seed, then we are in the situation studied in Section \ref%
{S: Game} and let $\mathbb{P}$ denote the resulting probability distribution
over $\Omega .$ Proposition \ref{P: small} (Case 1) implies that

\begin{claim}
With a single seed, if $c<\mathbb{P}\left[ G\mid N_{1}\right] ,$ then there
is an equilibrium in which everyone adopts regardless of information.
\end{claim}

The equilibrium payoff of an agent in the equilibrium of the claim is%
\begin{equation}
\pi =\rho -c  \label{single payoff}
\end{equation}

Finally, note that when $\varepsilon $ is small, 
\[
\mathbb{P}_{M}\left[ G\mid N_{1}\right] <\mathbb{P}\left[ G\mid N_{1}\right]
<\mathbb{P}_{M}\left[ Y_{2}\mid Y_{1}^{0}\right] 
\]

So if $c$ is such that%
\[
\mathbb{P}_{M}\left[ G\mid N_{1}\right] <c\leq \mathbb{P}\left[ G\mid N_{1}%
\right] 
\]%
the conditions of both claims above are satisfied. From (\ref{mult payoff})
and (\ref{single payoff}) 
\[
\pi _{M}-\pi =\rho \left( 1-\varepsilon \right) ^{2}\left( 1+2\varepsilon
\right) \left( 1-c\right) +\rho \left( 1-\varepsilon \right) \varepsilon
^{2}\left( -c\right) -\left( \rho -c\right) 
\]
and it may be verified that when $c\approx \mathbb{P}_{M}\left[ G\mid N_{1}%
\right] ,$ the difference is negative.

\subsection{Random seeds\label{A: Random}}

Denote by $\mathbb{P}^{1}$ the probability distribution over $\Omega $ when $%
1$ is the seed, by $\mathbb{P}^{2}$ the probability distribution over $%
\Omega $ when $2$ is the seed, and by $\widetilde{\mathbb{P}}$ the
probability distribution over $\Omega $ when the seeds are randomly chosen

Consider agent $1$ when uninformed. The probability that this agent assigns
to $G$ is%
\begin{eqnarray*}
\widetilde{\mathbb{P}}\left[ G\mid N_{1}\right] &=&\frac{\rho \left( \frac{1%
}{2}\varepsilon +\frac{1}{2}\left( \varepsilon +\left( 1-\varepsilon \right)
\varepsilon \right) \right) }{\rho \left( \frac{1}{2}\varepsilon +\frac{1}{2}%
\left( \varepsilon +\left( 1-\varepsilon \right) \varepsilon \right) \right)
+1-\rho } \\
&=&\frac{\frac{1}{2}\rho \varepsilon \left( 3-\varepsilon \right) }{\frac{1}{%
2}\rho \varepsilon \left( 3-\varepsilon \right) +1-\rho }
\end{eqnarray*}%
To see why, note that the numerator, the probability of $G\cap N_{1}$,
involves three possibilities in state $G$: (i) $1$ is the seed (probability $%
\frac{1}{2}$) and the message from the planner to $1$ was lost (probability $%
\varepsilon $); (ii) $2$ is the seed and the message from the planner to $2$
was lost; and (iii) $2$ is the seed, the message from the planner to $2$ was
received but then lost when sent to $1$. The denominator takes into account
the fact that when $\theta =b,$ no messages are sent.

Also, recall from (\ref{Pr[G given N1]}) that 
\[
\mathbb{P}^{1}\left[ G\mid N_{1}\right] =\frac{\rho \varepsilon }{\rho
\varepsilon +1-\rho } 
\]%
and note that $\mathbb{P}^{1}\left[ G\mid N_{1}\right] <\widetilde{\mathbb{P}%
}\left[ G\mid N_{1}\right] .$

Let $\varepsilon $ be small enough so that $\widetilde{\mathbb{P}}\left[
G\mid N_{1}\right] <\mathbb{P}^{1}\left[ Y^{\ast }\mid Y_{1}\right] .$

\begin{claim}
If 
\[
\mathbb{P}^{1}\left[ G\mid N_{1}\right] <c<\widetilde{\mathbb{P}}\left[
G\mid N_{1}\right] 
\]%
then with random seeding, there exists an equilibrium in which both agents
adopt regardless of whether they are informed or not.
\end{claim}

\begin{proof}
Suppose $2$ always adopts. Then the only uncertainty that $N_{1}$ faces is
regarding $G$ and since $c<\widetilde{\mathbb{P}}\left[ G\mid N_{1}\right] ,$
it is optimal for $N_{1}$ to also adopt. Then it is certainly optimal for $%
Y_{1}$ to adopt as well.

Agents $1$ and $2$ are symmetrically placed and so the claim is
established.\bigskip
\end{proof}

When $1$ is the only seed, Proposition \ref{P: small} shows that when $%
\mathbb{P}^{1}\left[ G\mid N_{1}\right] \leq c<\mathbb{P}^{1}\left[ Y^{\ast
}\mid Y_{1}\right] ,$ it is \emph{not} an equilibrium for both agents to
adopt regardless of whether they are informed or not.

Thus, with random seeding, the probability that everyone adopts is greater
than with a single seed. The payoffs of the two agents are also greater with
random seeding.


\begin{thebibliography}{99}
\bibitem{Akbarpour} Akbarpour, M., S. Malladi and A. Saberi (2023): "Just a
Few Seeds More: The Inflated Value of Network Data for Diffusion,"
https://web.stanford.edu/\symbol{126}mohamwad/

\bibitem{BBCG} Banerjee, A., E. Breza, A. Chandrasekhar and B. Golub (2023):
"When Less is More: Experimental Evidence on Information Delivery During
India's Demonetization," \emph{Review of Economic Studies,} forthcoming.

\bibitem{BCDJ1} Banerjee, A., A. Chandrasekhar, E. Duflo and M. Jackson
(2013): "The Diffusion of Microfinance," \emph{Science}, 341, 364--370.

\bibitem{BCDJ2} Banerjee, A., A. Chandrasekhar, E. Duflo and M. Jackson
(2019): "Using Gossips to Spread Information: Theory and Evidence from Two
Randomized Controlled Trials," \emph{Review of Economic Studies}, 86,
2453--2490.

\bibitem{Beaman} Beaman, L., A. B. Yishay, J. Magruder and A. Mushfiq
Mobarak (2021): "Can Network Theory-Based Targeting Increase Technology
Adoption?," \emph{American Economic Review}, 111, 1918--1943.

\bibitem{Blackwell} Blackwell, D. (1951): "Comparison of Experiments," in 
\emph{Proceedings of the Second Berkeley Symposium on Mathematical
Statistics and Probability}, Volume 1, University of California Press,
93--102.

\bibitem{Chetty} Chetty, R., J. Friedman and E. Saez (2013): "Using
Differences in Knowledge across Neighborhoods to Uncover the Impacts of the
EITC on Earnings," \emph{American Economic Review}, 103, 2683--2721.

\bibitem{Coles} Coles, P. and R. Shorrer (2012): "Correlation in the
Multiplayer Electronic Mail Game," \emph{BE Journal of Theoretical Economics}%
, 12.

\bibitem{Crouzet} Crouzet, N., A. Gupta and F. Mezzanotti (2023): "Shocks
and Technology Adoption: Evidence from Electronic Payment Systems," \emph{%
Journal of Political Economy}, 131, 3003--3065.

\bibitem{De Jaegher} De Jaegher, K. (2015): "Beneficial Long Communication
in the Multiplayer Electronic Mail Game," \emph{American Economic Journal:
Microeconomics}, 7, 233--251.

\bibitem{George} George, S., S. Gupta and Y. Neggers (2019): "Coordinating
Voters Against Criminal Politicians: Experimental Evidence from India,"
https://sites.google.com/view/siddharthgeorge/research

\bibitem{Gottlieb} Gottlieb, J. (2016): "Common Knowledge and Voter
Coordination: Experimental Evidence from Mali." In \emph{Voting Experiments}%
, edited by A. Blais, J-F. Laslier and K. Van der Straeten, New York:
Springer Press.

\bibitem{Jackson} Jackson, M. (2008): \emph{Social and Economic Networks},
Princeton: Princeton University Press.

\bibitem{KM} Kajii, A. and S. Morris (1997): "The Robustness of Equilibria
to Incomplete Information," \emph{Econometrica}, 65, 1283--1309.

\bibitem{Monderer} Monderer, D., and D. Samet (1989): "Approximating Common
Knowledge with Common Beliefs," \emph{Games and Economic Behavior}, 1,
170--190.

\bibitem{MondShap} Monderer, D., and L. Shapley (1996): "Potential Games," $%
\emph{Games}$ $\emph{and}$ $\emph{Economic}$ $\emph{Behavior}$, 124--143.

\bibitem{Oyama} Oyama, D. and S. Takahashi (2020): "Generalized Belief
Operator and Robustness in Binary-Action Supermodular Games," \emph{%
Econometrica}, 88, 693--726.

\bibitem{Rubinstein} Rubinstein, A. (1989): "The Electronic Mail Game:
Strategic Behavior under Almost Common Knowledge," \emph{American Economic
Review}\textbf{, }79, 385--391.

\bibitem{Sadler} Sadler, E. (2020): "Diffusion Games," \emph{American
Economic Review}, 110, 225--270.
\end{thebibliography}
\end{document}